%% file: main.tex
\newif\ifarxiv
\arxivtrue
\ifarxiv 
\documentclass[sigconf, nonacm, authorversion]{aamas} 
\fi

\ifarxiv \else
\documentclass[sigconf]{aamas} 
\fi

\usepackage{balance} % for balancing columns on the final page

\usepackage{balance} % for balancing columns on the final page

\usepackage{appendix}
 \usepackage{geometry}
 \usepackage{enumitem}
\usepackage{OPDLstyle}
\usepackage{amsmath}
\usepackage{caption}
\usepackage{subcaption}
\usepackage{multirow}
\usepackage{amsthm}
\usepackage{mathrsfs}
\usepackage{color}
\usepackage{graphicx}
\usepackage{ifthen}
\usepackage{xspace}
\usepackage{eurosym}

\usepackage{algorithm}
\usepackage{algpseudocode}
\algtext*{EndIf}
\algtext*{EndFor}

\algnewcommand\algorithmiccase{\textbf{case}}
\algdef{SE}[CASE]{Case}{EndCase}[1]{\algorithmiccase\ #1}{\algorithmicend\ \algorithmiccase}%
\algtext*{EndCase}
\makeatletter
\newenvironment{breakablealgorithm}
  {% \begin{breakablealgorithm}
   \begin{center}
     \refstepcounter{algorithm}% New algorithm
     \hrule height.8pt depth0pt \kern2pt% \@fs@pre for \@fs@ruled
     \renewcommand{\caption}[2][\relax]{% Make a new \caption
       {\raggedright\textbf{\fname@algorithm~\thealgorithm} ##2\par}%
       \ifx\relax##1\relax % #1 is \relax
         \addcontentsline{loa}{algorithm}{\protect\numberline{\thealgorithm}##2}%
       \else % #1 is not \relax
         \addcontentsline{loa}{algorithm}{\protect\numberline{\thealgorithm}##1}%
       \fi
       \kern2pt\hrule\kern2pt
     }
  }{% \end{breakablealgorithm}
     \kern2pt\hrule\relax% \@fs@post for \@fs@ruled
   \end{center}
  }
\makeatother

\allowdisplaybreaks

\usepackage{tikz}
\usetikzlibrary{arrows,positioning}
\tikzset{
mystyle/.style={
  circle,
  inner sep=0pt,
  align=center,
  draw=black,
  fill=white
  }
}

\ifarxiv \else
\acmSubmissionID{530}
\fi 
\title{Rational Capability  in Concurrent Games}
\ifarxiv \thanks{This is an extended version of the same title paper that will appear in AAMAS
2025. This version contains technical appendixes with proof details that, for space reasons, do
not appear in the AAMAS 2025 version.} 
\fi
\author{Yinfeng Li}
\affiliation{
  \institution{IRIT, CNRS, University of Toulouse}
  \city{Toulouse}
  \country{France}}
\email{yinfeng.li@irit.fr}

\author{Emiliano Lorini}
\affiliation{
  \institution{IRIT, CNRS, University of Toulouse}
  \city{Toulouse}
  \country{France}}
\email{emiliano.lorini@irit.fr}

\author{Munyque Mittelmann}
\affiliation{
  \institution{University of Naples Federico II}
  \city{Naples}
  \country{Italy}}
\email{munyque.mittelmann@unina.it}

\begin{abstract}
We extend  concurrent game structures (CGSs)
    with a simple  notion
    of preference
    over computations 
    and define 
    a minimal notion of rationality
    for agents based on the concept
    of dominance. We use this notion to interpret a $\cllogic$  
    and an $\atllogic$ languages  that extend the  basic  $\cllogic$ and $\atllogic$ languages with
    modalities for 
    rational  capability, namely,
    a coalition's capability to \emph{rationally} 
    enforce a given
    property. 
    For each of these languages, we provide 
    results
    about the complexity of satisfiability checking and model checking as well as about axiomatization. 
\end{abstract}

\keywords{Logics for Multi-Agent Systems; Rationality; Strategic Reasoning}

\newcommand{\BibTeX}{\rm B\kern-.05em{\sc i\kern-.025em b}\kern-.08em\TeX}

\newtheorem{theorem}{Theorem}

\newtheorem{example}{Example}
\newtheorem{lemma}{Lemma}
\newtheorem{corollary}{Corollary}

\newtheorem{definition}{Definition}

\newtheorem{fact}{Fact}

\newbox\itembox
\def\itemlistlabel#1{#1\hfill}
\def\itemlist#1{\setbox\itembox=\hbox{#1}%
                \list{}{\labelwidth\wd\itembox
                             \leftmargin\labelwidth
                             \advance\leftmargin by\itemindent
                             \advance\leftmargin by\labelsep
                             \let\makelabel\itemlistlabel}}

\DeclareMathAlphabet\mathbfcal{OMS}{cmsy}{b}{n}

\begin{document}

%%% The following commands remove the headers in your paper. For final 
%%% papers, these will be inserted during the pagination process.

\pagestyle{fancy}
\fancyhead{}

%%% The next command prints the information defined in the preamble.

\maketitle

\section{Introduction}

The field of 
logics for multi-agent systems
has been very active
in the last twenty years.
Different logics
have been proposed
and their proof-theoretic,
complexity
and algorithmic 
aspects 
for satisfiability
and model checking studied
in detail. 
The list of logics in this area is long.
It includes
alternating-time temporal logic
($\atllogic$) \cite{alur2002alternating,GorankoDrimmelenATL},
its ``next''-fragment 
called 
coalition logic
($\cllogic$) \cite{DBLP:journals/logcom/Pauly02,GorankoTARK2021},
the logic of agency $\stitlogic $ \cite{DBLP:journals/logcom/BroersenHT06,DBLP:conf/atal/BoudouL18}, and 
the more expressive strategy logic
($\strategylogic$) \cite{MogaveroMPV14,DBLP:conf/concur/ChatterjeeHP07}.
A  widely used
semantics for interpreting
these logics is based  on  concurrent game structures (CGSs), 
transition systems
in which state-transitions are labeled by joint actions of agents. 
A CGS allows us to represent the repeated interaction between multiple agents 
in a natural way as well as their choices and strategies.
It is similar
to the game-theoretic concept of dynamic 
game
in which players move sequentially or repeatedly.
But an element that is missing from  CGSs compared 
to dynamic games is the preference of the agents.
Indeed, most logics for multi-agent
systems
including 
$\atllogic$,
$\cllogic$,
$\strategylogic$
and 
 $\stitlogic $ abstract away
 from the agents' preferences
 as they are only interested in representing
 and reasoning about 
 the 
 \emph{game form}, namely, the way an outcome
 is determined based on the agents' 
 concurrent choices over time.

In this paper  we extend 
CGSs with a basic concept of preference.
This is in order to have a semantics that allows us to represent a game in its entirety, capturing both its aspects (the game form and the agents' preferences), and consequently to reason about %the agents' 
rational
choices and strategies 
in the game. 
Specifically, we introduce a new class
of structures called CGS with preferences
(CGSP) that includes  one  preference
ordering for each agent at  each state in the underlying CGS.
An agent's preference at a given state is relative to
the set of computations (or histories) starting at this state.
We consider an interesting subclass
of CGSP with \emph{stable preferences}
in which agents' preferences do not change over time. This reminds
the notion of
time consistency
of preferences studied in  economics, 
in opposition to time inconsistency \cite{Lowewenstein2002}. 
We employ   CGSP
to interpret
two novel
languages 
$\atlratlogic$
and 
$\clratlogic$ ($\atllogic$/$\cllogic$ with minimal rationality) 
  that extend the  basic   $\atllogic$ and 
  $\cllogic$ languages with
    modal operators  for 
    \emph{rational}   capability, namely,
    a coalition's capability to 
    enforce a given outcome by choosing a 
        \emph{rational} strategy. 
        The notion of rationality
        that we use to define these operators  
        is based on strong dominance: 
        the collective strategy of a coalition is rational insofar as the individual strategies that compose it are not strongly dominated.
It is a minimal notion of rationality since it does not require the agent to reason about what others will choose. It simply requires an  agent not to play a strategy that is beaten by another
of its strategies  regardless of what  the others choose.
In \cite{Horty2001}, it  is
shown that this minimal dominance-based  requirement of rationality
is particularly suitable for defining the deontic  notion of obligation, namely, what an agent or coalition
ought to do. 
The general idea of refining the capability operators of ATL by restricting quantification to the agents' rational strategies is shared  with 
Bulling et al. 
\cite{BullingJD08}.
But unlike us, they %Bulling et al. 
do not extend  CGSs 
with an explicit notion of preference. In their semantics 
sets of 
plausible/rational strategies 
can be only referred to via
atomic plausibility terms (constants) whose interpretation is ``hardwired''
in the model. 
A similar idea can also be found in \cite{DBLP:journals/sLogica/LoriniS16}
in which 
rational 
 $\stitlogic $ (``seeing to it that'')
modalities are introduced.

For each of the languages we introduce,  results about the complexity of satisfiability checking and model checking as well as about axiomatization
are provided.
In particular, the following are the main results of the paper:
\begin{itemize}
\item tree-like model property for $\atlratlogic$;  
    \item polynomial embeddings
    of $\atlratlogic$
    into 
    $\atllogic$
    under the stable preference assumption,
    and 
    of $\clratlogic$
    into 
    $\cllogic$
    both 
       under the stable preference assumption
       and with no assumption; 

       \item thanks to the embeddings, 
     tight complexity results
     of satisfiability checking
     for  $\atlratlogic$
    and
$\clratlogic$; 

\item a sound and complete axiomatization
for the logic 
$\clratlogic$; 

\item a model checking algorithm
for $\atlratlogic$ for the class of concurrent game structures with short-sighted preferences.

\end{itemize}

    The paper is organized as follows. In Section
    \ref{sec:relatedwork},
    we discuss related work.
    In Section \ref{sec:semantics},
    we present the semantic foundation
    of our framework. 
    Then, 
    in Section \ref{sec:language}, 
    we introduce the languages of 
    $\atlratlogic$
    and
$\clratlogic$.
Section \ref{sec:resultsblock}
is devoted
to the tree-like model property,
the embeddings
and the complexity results
for our logics. In Section \ref{axiomatization}
we deal with axiomatization, while in 
Section \ref{sec:modelchecking}
we move to model checking. 
\ifarxiv 
%This paper is the extended version of the paper with the same title accepted at AAMAS 2025. 
 Detailed proofs are given in   
 Appendixes A-H.
\else
\textcolor{blue}{Detailed proofs are presented in the extended version of this paper, available on ArXiV [add url]. }
\fi 

% The paper only contains some sketches of proof.
% Detailed proofs are given in the 
% technical
% annex at the end
% of the extended version of the paper provided as supplementary material.

 % In game theory and related fields, a game form, game frame, ruleset, or
 % outcome function is the set of rules that govern a game and determine its outcome based on each player's choices. A game form differs from a game in that it does not stipulate the utilities or payoffs for each agent

 % about  what an agent
 % (resp. coalition) actually does and can do
 % by choosing a certain individual 
 % (resp. joint)
 % action  or individual (resp. collective)  strategy. 

\section{Related work}\label{sec:relatedwork}

%\begin{itemize}   \item Work of Wooldrige on Rational Verification  [Munyque / Emiliano], + Lexicographical preferences
%    \item Check Naumov's papers, e.g., "A logic of higher-order preferences" and more [Yinfeng]
%\item Discuss papers on SL and ATL using quantitative semantics [Munyque]
%\item Discuss  Horty's STIT deontic logic of obligations based on dominance [Emiliano]
%\item Lorini et Liu's papers on logics of preferences.  and recent Grossi and van der Hoek's paper on the general logic for preferences 
%Thomas Ågotnes: "On the logic of preference and judgment aggregation", and "Reasoning about reasons behind preferences using modal logic"\end{itemize}

Several works have studied logics  for reasoning about preferences, without considering the strategic and temporal dimensions. In particular, \citeauthor{BenthemL07} \cite{BenthemL07} proposed a dynamic logic of knowledge update and preference upgrade, where incoming suggestions or commands change the preference relations. \citeauthor{Lorini21} \cite{Lorini21} presented a general logical framework for reasoning about agents’ cognitive
attitudes, which captures concepts of knowledge, belief, desire, and preference.  
%Logic of preference  
\citeauthor{GrossiHK22} \cite{GrossiHK22} investigated four different semantics for conditional logics based on preference relations over alternatives. The semantics differ in the way of selecting the most preferred alternative, which includes maximality, optimality, unmatchedness, and acceptability. \ifarxiv Maximality is related to the rationality concept we consider in this paper. Maximal alternatives are those not strictly dispreferred to any other. \fi 

Two of the most important developments in logics for strategic reasoning are ATL \cite{alur2002alternating} and SL \cite{MogaveroMPV14}.
Unlike ATL, SL can express complex solution concepts (such as dominant strategy equilibrium) and thus capture some notions of rationality.  
However, in both logics, agents' preferences are not modeled intrinsically, instead, their goals can be represented as Boolean formulas. 
A way to incorporate preferences in those logics is to include atomic propositions stating that the utility of an agent is greater than or equal to a given value \cite{baltag2002logic}, which requires an exhaustive enumeration for each relevant utility threshold.  
The extensions of ATL and SL with quantitative semantics \cite{jamroga2024playing,bouyer2023reasoning} generalize fuzzy temporal logics and capture quantitative goals.  This approach has been recently used to represent agents' utilities in  mechanism design ~\cite{SLKF_KR21,MittelmannMMP22}. 

The dominance relation among strategies has been considered alongside specifications in temporal logics  \cite{AminofGR21,AminofGLMR21}. These works provide algorithms for synthesizing best-effort strategies, which are maximal in the dominance order, in the sense that they achieve the agent goal against a maximal set of environment specifications.

Rationality in concurrent games is typically associated with  a\-gents' knowledge and preferences. Know-How Logic with the Intelligence \cite{naumov2021intelligence} captures rational agents' capabilities that depend on the intelligence information about the opponents’ actions. The interplay between agents' preferences and their knowledge was described in \cite{Naumov2023AnEL}. 
A sound, complete, and decidable logical system expressing higher-order preferences to the other agents was given in  \cite{Jiang2024-JIAALO}. However, none of these three papers address the connection between rational agents' capabilities and their preferences.

Our work is also related to the research on rational verification and synthesis. The first is the problem of checking whether a temporal goal is satisfied in some or all game-theoretic equilibria of a CGS \cite{AbateGHHKNPSW21,GutierrezNPW23}. Rational synthesis consists in the automated construction of such a model \cite{FismanKL10, CFGR16}. 
Different types of agent objectives have been considered, including Boolean temporal specifications \cite{gutierrez2019equilibrium}, mean payoff \cite{gutierrez2024characterising}, and lexicographical preferences \cite{gutierrez2017nash}

While being able to analyze multi-agent systems with respect to solution concepts, both rational verification and model-checking SL specifications face high complexity issues. 
In particular, key decision problems for rational verification with temporal specifications are known to be \DExptime-complete \cite{GutierrezNPW23} and model-checking  SL is non-elementary for memoryful agents \cite{MogaveroMPV14}.

ATL with plausibility \cite{BullingJD08} allows the specification of sets of
rational strategy profiles, and reason about agents' play if the agents can only play  these strategy profiles. The approach considers plausibility terms, which are mapped to a set of strategy profiles.  
The logic includes formulas of the form $(\text{set-pl} \omega) \varphi$, meaning  that ``assuming that the set of rational strategy profiles is defined in terms of the plausibility terms $\omega$, then, it is
plausible to expect that $\varphi$ holds''. %This is similar to rational verification, as it assumes that players play rationally with respect to some solution concept, and analyze the game outcome 
% under this assumption. 
This idea was extended in  \cite{BullingJ09} to a variant of SL for imperfect information games.
However, as emphasized in the introduction, 
Bulling et al. do not 
represent 
agents' 
preferences in their  semantics. 
This is a crucial
difference between their work and ours.
Our main focus is
on  extending CGSs
with preferences,
studying the dynamic properties
of agents' preferences in concurrent games,
and defining a logic
of rational capability with the help
of the semantics
combining CGSs with preferences.

\section{Semantics}\label{sec:semantics}

In this section,
we first  define the basic
elements of the semantics:
the notions
of concurrent game structure (CGS),
computation
and strategy.
Then, we 
extend a CGS
with preferences
and use the resulting
structure to define the notion
of dominated strategy.

\subsection{Preliminaries}\label{sec:sempreliminaries}

Let $\ATM$  be a countable set of atomic propositions
and $\AGT=\{1, \ldots ,n\}$ a finite set of agents.
A coalition is a (possibly empty)
set of agents
from $\AGT$. 
Coalitions are denoted by
$\Group, \Group', \ldots  $
$\AGT$
is also called the grand coalition. 
The following definition introduces
the concept
of concurrent game
structure (CGS), as defined in \cite{{DBLP:conf/atal/BoudouL18}}. 
 \begin{definition}[CGS]\label{CGS}
A concurrent game structure (CGS)
is a tuple 
$M  = \big( W,  \ACT, (\relAct{\jactatm})_{\jactatm \in\JACT},  \valProp  \big)$ 
with 
\begin{itemize}
\item $W$ a non-empty set of worlds
(or states),
\item $\ACT$ a set of action names
and  
$\JACT = \ACT^n      $  the corresponding  set of
joint action names,

\item  $\relAct{\jactatm} \subseteq W \times W$ a  transition relation
for joint action $\jactatm$, 
\item $\valProp: W \longrightarrow 2^\ATM$
a valuation function, 
\end{itemize}
 such that
for every $w \in W$ and $\jactatm \in\JACT $:

\begin{enumerate}[label=(C\arabic*)]

\item\label{C1}  $ \relAct{\jactatm}  $ is deterministic (\emph{collective choice determinism}),\footnote{A relation $ \relAct{}  $ is   deterministic if
$\forall w,v,u \in W, \text{ if } w\relAct{} v\text{ and }w   \relAct{}  u \text{ then }v=u $.
}

\item\label{C2}  if $ \jactatm(1) \in \choiceSet_1(w), \ldots , \jactatm(n) \in \choiceSet_n( w)$
then $\relAct{ \jactatm  } (w) \neq \emptyset$ (\emph{inde\-pendence of choices}),

\item\label{C3}   $\relAct{ }  $
is serial (\emph{neverending interaction}),\footnote{A relation $ \relAct{}  $ is  serial if
$\forall w \in W , \exists v \in W \text{ s.t. }w\relAct{}v  $.
}

\end{enumerate}
where 
\begin{align*}
\relAct{ }  & = \bigcup_{  \jactatm \in \JACT }  \relAct{  \jactatm},\\
\choiceSet_i( w) & = \{  \actatm \in \ACT \suchthat \exists \jactatm  \in \JACT \text{ s.t. } \relAct{ \jactatm  } (w) \neq \emptyset \text{ and } \jactatm(i)= \actatm  \}.
\end{align*}
%is agent $i$'s set of available choices (or moves) at $w$.

  % The class of CGSs with stable preferences is denoted
  %         by $\classcgs$.
\end{definition}

% A CGS is a  type of multi-relational
%  structure, \emph{alias} Kripke model,
%  the kind of structure traditionally
%  used in modal
%  logic.

 The previous definition  slightly differs
from the usual definition of
CGS used for interpreting 
$\atllogic$ \cite{GorankoDrimmelenATL} and 
strategy logic ($\strategylogic$) \cite{MogaveroMPV14}. 
In particular a
CGS, as defined 
in Definition \ref{CGS},
is a 
 multi-relational
  structure, \emph{alias} Kripke model,
  the kind of structure traditionally
  used in modal
  logic.
  Every joint action
  is associated 
to a binary relation
over states satisfying certain properties, 
while 
in the usual
semantics for 
$\atllogic$
and 
$\strategylogic$
 a transition function
 is used that 
maps a state
and a joint action executable
at this state to a successor state.
The two variants  are interdefinable.
We use the multi-relational variant  of CGS
since it
is particularly convenient for proving the model-theoretic
and proof-theoretic results
in the rest of the paper.

The relation 
$\relAct{\jactatm } $
with $\jactatm \in \JACT$
is used to identify  the set of states
$\relAct{\jactatm }(w)=
\{v \in W :
w \relAct{\jactatm } v \}$
that are reachable from
state $w$
when  the agents collectively choose
joint action 
$\jactatm$
at state $w$,
that is, 
when every agent $i$
chooses
the individual component $\jactatm(i)$
at state $w$. 
 $\relAct{\jactatm } (w)=\emptyset$
 means that the  joint action $\jactatm  $ cannot
 be collectively chosen by the agents
 at state $w$. 
 The set
 $\choiceSet_i( w)$
 in the previous definition
 corresponds
 to
 agent $i$'s  choice set
 at state $w$, i.e.,
  the set of actions
 that agent $i$
 can choose at state $w$
 (or agent $i$'s set of available actions
 at $w$). Note that an agent's choice
set may vary from one state to another,
i.e.,
it might be the case that
 $\choiceSet_i( w)\neq \choiceSet_i( v)$
 if  $w\neq v$. 
Constraint C1 captures 
\emph{collective choice  determinism}:
the outcome
of a collective choice of all agents is uniquely determined.
Constraint C2 corresponds to the
\emph{independence of  choices} assumption:
if agent $1$
can individually  choose action 
$ \jactatm(1)$, 
agent $2$
can individually  choose action 
$ \jactatm(2)$,...,
agent $n$
can individually  choose action 
$ \jactatm(n)$, 
then the agents can collectively choose joint action 
$\jactatm$. More intuitively, this means that agents can never be deprived of choices due to the choices made by other agents.
Constraint C3 corresponds to the 
\emph{neverending interaction} assumption:
every state in a CGS has
 \emph{at least one successor},
 where the successor of a given state
 is a state
 which is reachable
 from
 the former via
 a collective choice
 of all agents.

 For notational
 convenience,
 in the rest
of the paper,
sometimes 
use the 
abbreviation 
 $\relActJoint \eqdef   (\relAct{\jactatm})_{\jactatm \in\JACT}$
 to indicate a profile
 of transition relations,
 and write
 $M  = ( W,  \ACT, 
 \relActJoint,  \valProp  )$ 
 instead
 of $M  = \big( W,  \ACT, (\relAct{\jactatm})_{\jactatm \in\JACT},  \valProp  \big)$ 
 for a CGS.

\begin{example}[Crossing road]
Assume a model $M_{cross}$ %CGSP $P_{cross} = (M_{cross},\Omega_{M_{cross}})$
representing a system with two  vehicles (denoted $v_1$ and $v_2$) that need to decide how to act when approaching  intersections. Each vehicle can either go straight on ($Move$) or wait ($Skip$). Their %vehicles'
goal is to cross the road, but they prefer to avoid collisions, which happen when they go straight at the same time. $M_{cross}$ is represented by Figure \ref{fig:model}. 
The initial state is denoted with $init$, while $crash$ denotes the failure state (i.e., a collision occurred). The proposition $c_1$ (similarly, $c_2$) indicates the situation in which the vehicle $v_1$ has crossed (resp., $v_2$).

\input{figModel}

%Formally, $P_{cross} = (M_{cross},\Omega_{M_{cross}})$, where $M_{cross}$ is represented by Figure \ref{fig:model}
%$M_{cross} = ( W,  \ACT,  (\relAct{\jactatm}   )_{\jactatm \in \JACT},  \valProp  )$, $\Omega_M=(\preceq_{i,w }   )_{i \in \AGT, w \in W }$ and
%\begin{itemize}
%    \item $W =\{w_0, w_1, w_2, w_4\}$ 
%    \item $\ACT = \{Move, Skip\}$ 
%    \item $\relAct{(Skip,Skip)} = \{(i,i):i \in W\}$ ; 
%    \item $\relAct{(Skip,Move)} = \{(w_0,w_1),
%    (w_1,w_1),     (w_2,w_3), (w_3, w_3), (w_4,w_4) \}$;  
%    \item $\relAct{(Move, Skip)} = \{(w_0,w_2),
%    (w_1,w_3),     (w_2,w_2), (w_3, w_3), (w_4,w_4) \}$;  
%    \item $\relAct{(Move, Move)} = \{(i,w_4): i \in W \}$;  %->changed this in the figure to no collision after crossing
% \item $ \valProp(w_0)=\emptyset;   \valProp(w_1)=\{crossed_{1}\};     \valProp(w_2)=\{crossed_{2}\};    \valProp(w_3)=\{crossed_{1},crossed_{2}\}; \valProp(w_4)=\{collision\}   $
%\end{itemize}
 %$\preceq_{i,w } $ for each $i \in \AGT, w \in W$  \textcolor{red}{( they prefer paths with "no collision and they crossed" > "no collision", they can also prefer to cross the sooner)}

\end{example}

The following definition introduces the notions
of path and computation, two essential elements of temporal
logics and logics for strategic reasoning.
 \begin{definition}[Path and computation]
A \emph{path}
in a  CGS
$M = ( W,  \ACT,  \relActJoint,  \valProp  )$ is a sequence
$\lambda= w_0 w_1 w_2 \ldots$
of states from  $W$ 
 such that
$w_k    \relAct{ } w_{k +1 }$
for all $k \geq 0$,
    where we recall $\relAct{ }   = \bigcup_{  \jactatm \in \JACT }  \relAct{  \jactatm}$. 
The set of all paths in $M$ is denoted by  $\pathset_M$.
Given a path $\lambda$
of length higher than $k'$ and $k\leq k'$,
the $k$-th element of $\lambda$
is denoted by $\lambda(k)$. 
A \emph{computation}  (or  \emph{full path})
in $M$
is a path $\lambda \in \pathset_M$
such that there is no $\lambda' \in \pathset_M$ of which $\lambda $ is a proper prefix.
The set of all computations in $M$ is denoted  by $\historyset_M$.
The set of all computations in $M$
starting at world $w \in W$ 
(i.e., whose first element is $w$) 
is denoted by  $\historyset_{M,w}$.
\end{definition}

From Constraint C3 in Definition \ref{CGS}, 
it is easy to prove the following fact. 
\begin{fact}
    If $\lambda \in\historyset_M $
    then $\lambda$
    is infinite.
\end{fact}

%%Due to Constraint C3 in Definition 1 (i.e., seriality), a computation can be explicitly defined as an infinite path without using the notion of proper prefix. The current definition of computation given in Definition 2 works even without Constraint C3. 

An agent's individual 
perfect recall
strategy is nothing
but the specification
of a choice for the agent at the end
of every finite path 
in a CGS. 
It is formally defined as follows. 

\begin{definition}[Individual strategy]
Let $M = ( W,  \ACT,  \relActJoint,  \allowbreak \valProp  )$
be a CGS. A (perfect recall) strategy for agent  $i$
    in $M$
    is a function $\strategy_i   $
    that maps every
    finite path $w_0 \ldots w_k\in 
    \pathset_{M }$
    to a choice $\strategy_i(w_0 \ldots w_k) \in \choiceSet_i( w_k )$
    available to agent $i$
    at the end of this finite path,
    where again  we recall $\relAct{ }   = \bigcup_{  \jactatm \in \JACT }  \relAct{  \jactatm}$.
\end{definition}

A collective  strategy   is the assignment of an individual
strategy to each agent.

\begin{definition}[Collective strategy]
Let $M = ( W,  \ACT, \relActJoint, \allowbreak \valProp  )$
be a CGS. 
A collective strategy   for a  coalition $\Group$
in $M $
     is a function $\strategymap_\Group$
     that associates every agent $i \in \Group$
     to a strategy $\strategymap_\Group(i)$
     for $i$ in $M$. 
     The set of collective  strategies for coalition
      $\Group$ in $M$ is denoted by  $\stratsetatl_{M}^\Group $.
      Its elements are denoted by 
      $\strategymap_\Group , \strategymap_\Group', \ldots $
%      We write
%        $\strategy$ instead of $\strategy_\AGT$,
%        and
%  $\stratsetatl$ instead of  $\stratsetatl_\AGT$.
%For terminological convenience, 
%elements of $\stratsetatl_{M}^\AGT $
%are called complete (collective)  strategies. 
%We define 
% $\stratsetatl_{M}= \bigcup_{\Group \subseteq \AGT }\stratsetatl_{M}^\Group $
%to be the set of all collective  strategies in $M$
%and note $\strategymap, \strategymap', \ldots $ its elements.
%Given an arbitrary  collective strategy
%$\strategymap \in \stratsetatl_{M} $
%we note $\mathit{dom}(\strategymap)$
%its domain. 
\end{definition}

Given a coalition $\Group$, 
$\strategymap_\Group \in \stratsetatl_{M}^\Group $
and 
$\strategymap_{\AGT \setminus \Group}' \in \stratsetatl_{M}^{\AGT \setminus \Group} $,
we define 
$\strategymap_\Group 
\oplus 
\strategymap_{\AGT \setminus \Group}'
\in \stratsetatl_{M}^{\AGT } 
$ to be the composition of the two strategies:
\begin{align*}
&\strategymap_\Group 
\oplus 
\strategymap_{\AGT \setminus \Group}'(i)=
\strategymap_\Group (i) \text{ if }i \in \Group,\\
&\strategymap_\Group 
\oplus 
\strategymap_{\AGT \setminus \Group}'(i)=
\strategymap_{\AGT \setminus \Group}'(i) \text{ otherwise}.
\end{align*}

%\circ to \oplus

Given an initial
state $w$ and a  collective strategy  for
a coalition 
  $\Group$
  we can compute the set of computations generated by
  this strategy
  starting at $w$.

\begin{definition}[Generated computations]
Let $M = ( W,  \ACT, $ $ \relActJoint,  \valProp  )$
be a CGS,
$w\in W$
and $\strategymap_\Group \in \stratsetatl_{M}^\Group  $. 
The set 
$\outset_M(w, \strategymap_\Group)$
denotes
the
set of all computations
$\lambda = w_0 w_1 w_2 \ldots $
in $\historyset_{M}$
such that
$w_0=w$
and
for every $k \geq 0$,   there is $ \jactatm \in \JACT  $ such that:
\begin{itemize}
%\item
%$\relAct{ \jactatm  } (w_k ) \neq \emptyset$,
\item $\strategymap_\Group (i)(w_0 \ldots w_k) =  \jactatm(i) $
for all $i \in \Group$, and 
\item $ w_k  \relAct{ \jactatm  } w_{k +1  }$. 
\end{itemize}

\end{definition}

$\outset_M(w, \strategymap_\Group)$
is
the set of computations in  $M$
generated by 
coalition $\Group$'s
collective strategy 
$\strategymap_\Group$
starting at state $w$. 
 Note that
the set $\outset_M (w, \strategymap_\AGT )$
is a singleton 
because of 
Constraint C1  
for collective choice  determinism. 
The unique element of 
$\outset_M (w, \strategymap_\AGT )$
is denoted by $\lambda^{M{,}w{,}\strategymap_\AGT }$.

Note also that there is a single strategy 
$\strategymap_\emptyset $
for the empty coalition, the one which
makes no assignments at all.
Thus, 
$\outset_M (w, \strategymap_\emptyset  )=\historyset_{M,w}$. 

\subsection{Adding preferences}
\label{sec:pref}

In this section,
we extend the notion
of CGS of Definition \ref{CGS}
with preferences.
 \begin{definition}[CGS with preferences]\label{def:cgspref}
 Let 
 $M = ( W,  \ACT,  \allowbreak \relActJoint, 
  \valProp  )$
  be a CGS. A preference structure for $M$
  is a tuple $\Omega_M=(\preceq_{i,w }   )_{i \in \AGT, w \in W }$
  where, for every $i \in \AGT$ and $w \in W$, 
  $\preceq_{i,w }$
  is total preorder over $  \historyset_{M,w}$.
  We call the pair $(M,\Omega_M)$
  a CGS with preferences (CGSP). 
  As usual, we write 
    $\lambda' \prec_{i,w } \lambda$
    if 
       $\lambda' \preceq_{i,w } \lambda$
       and
          $\lambda \not \preceq_{i,w } \lambda'$.
          % The class of CGSPs is denoted
          % by $\classcgsp$.
          
          We say that the CGSP 
          $(M,\Omega_M)$ has stable preferences
          if
          the following condition holds:
          \begin{align*}
       (\mathbf{SP}) \   \forall w,v \in W ,
          \forall \lambda, \lambda' \in 
          \historyset_{M,v },
          & \text{ if } w \relAct{ }  v
          \\ &
          \text{ then }
          \big(
        \lambda'   \preceq_{i,v}\lambda  \text{ iff }
        w  \lambda'   \preceq_{i,w } w\lambda 
          \big). 
          \end{align*}
\end{definition}
Constraint $ \mathbf{SP}$
for stable preferences captures the fact
that an agent's preference is stable over time:
an agent prefers a computation $\lambda $ to a computation $\lambda '$ 
starting  at  the same world  $v$ 
if and only if it prefers the precursor  of $\lambda$
(i.e., $w \lambda $)
to the precursor of $\lambda '$ 
(i.e., $w \lambda'  $)
at each
predecessor $w$ of $v$.

\begin{example}[Crossing road (cont.)]
Let us resume our example. 
The preference relations $\preceq_{v_1,w_0}$ and $\preceq_{v_2,w_0}$ of agents $v_1$ and $v_2$ (resp.) in state $w_0$ is illustrated in Figure \ref{fig:pref} (preference relation over the other states are analogous). 
The intuition of the preference of each agents $v_i$ is that the less preferred situation for each agent is when there is a collision (the computation indicated with $-_i$). 
Additionally, the agents prefer computations in which he crossed (indicated by %$++$ and
$+_i$) to the ones he did not ($=_i$). 
%Finally, the agent prefers to cross before the agent $v_2$ (computations indicated with $++$). 
We denote by
$P_{cross} = (M_{cross},\Omega_{M_{cross}})$ the CGS $M_{cross}$ with preferences 
$\Omega_M=(\preceq_{i,w } )_{i \in \AGT, w \in W }$.

\input{figPref}

The strategies in which the agent performs $Move$ in the initial state are not dominated, because it may lead to the state where he crossed or to a collision. 
On the other hand, the strategy in which the agent waits (action $Skip$) when only the other agent has crossed is dominated by the strategy in which he moves whenever agent $v_2$ has crossed.   
\end{example}

The following definition introduces
the notion 
of dominated strategy,
the essential constituent 
of minimal rationality
for agents. 
 \begin{definition}[Dominated strategies]
Let $P=(M,\Omega_M)$ be a CGSP
with  $M = ( W,  \ACT,  \relActJoint, 
  \valProp  )$
  a CGS and $\Omega_M=(\preceq_{i,w }   )_{i \in \AGT, w \in W }$
  a preference structure for $M$, 
 $i \in \AGT$, $w \in W$,
  and  $\strategymap_{\{i\}},\strategymap_{\{i\}}' \in \stratsetatl_{M}^{\{i\}}$.
  We say that at world $w$
  agent $i$'s strategy $\strategymap_{\{i\}}' $
  dominates agent $i$'s  strategy $\strategymap_{\{i\}} $
  iff 
  \begin{align*}
  \forall \strategymap_{\AGT \setminus {\{i\}}}''
    \in \stratsetatl_{M}^{\AGT \setminus {\{i\}} },
 \lambda^{M{,}w{,}\strategymap_{\{i\}}   \oplus 
 \strategymap_{\AGT \setminus {\{i\}} }''}  
 \prec_{i,w}  \lambda^{M{,}w{,}\strategymap_{\{i\}}'  \oplus 
 \strategymap_{\AGT \setminus {\{i\}}}''}  . 
  \end{align*}
  Agent $i$'s strategy $\strategymap_{\{i\}} $
    is said to be dominated at $w$ 
    if
    there exists another strategy $\strategymap_{\{i\}} '$
    of $i$
    which dominates $\strategymap_{\{i\}} $  at $w$.
 Agent $i$'s set of dominated strategies
  at $w$ 
  is denoted by $\mathit{Dom}_{M,w}^i$.
  % Conversely, $\mathit{Undom}_{M,w}^i$
  %  is  agent $i$'s set of non-dominated strategies
  % at $w$. 
\end{definition}

In the next section
we introduce 
a novel
language
that extends
the language
of $\atllogic$
with a family of operators
for rational capability.
It will be  interpreted by means
of the notion of CGSP.

\section{Language }\label{sec:language}

The language of $\atlratlogic$
($\atllogic$ with \emph{minimal rationality}), denoted 
by $\lang_{\atlratlogic}(\ATM, \AGT )$,
is defined by the following grammar: 
\begin{center}\begin{tabular}{lcl}
 $\varphi, \psi$  & $\bnf$ & $p \mid \neg\varphi \mid \varphi \wedge \psi  \mid
 \atlop{\Group}
\nexttime \varphi \mid 
 \atlop{\Group}
\henceforth  \varphi 
 \mid   \atlop{\Group}
 (\until{\varphi}{\psi} )$\\
 &  & $ \atloprat{\Group}
\nexttime \varphi \mid 
 \atloprat{\Group}
\henceforth  \varphi 
 \mid   \atloprat{\Group}
 (\until{\varphi}{\psi} ),$
\end{tabular}\end{center}
where  $p$
ranges over  $\ATM$
and $\Group $ ranges over $2^\AGT$. 
The other Boolean connectives
and constructs 
$\vee, \rightarrow, \leftrightarrow, \top, \bot $
are defined as abbreviations in the usual way.

On the one hand, 
formulas 
$  \atlop{\Group}
\nexttime \varphi$,
$  \atlop{\Group}\henceforth \varphi  $
and 
$  \atlop{\Group}( \until { \varphi    } { \psi    }) $
capture the notion of strategic  capability.
They 
have the usual
$\atllogic$ readings:
$  \atlop{\Group}
\nexttime \varphi$
has to be read 
``coalition $\Group$ has 
a strategy at its disposal 
which guarantees that
$\varphi$ is going to be true in the next state'',
while 
$  \atlop{\Group}\henceforth \varphi  $
has to be read 
``coalition $\Group$ has 
a strategy at its disposal 
which guarantees that   
 $\varphi$ will always be true''. 
Finally, 
$  \atlop{\Group}( \until { \varphi    } { \psi    }) $
has to be read 
``coalition $\Group$ has 
a strategy at its disposal 
which guarantees that   
 $\varphi$ will be true until $\psi$ is true''.
 On the other hand,
 formulas 
 $  \atloprat{\Group}
\nexttime \varphi$,
$  \atloprat{\Group}\henceforth \varphi  $
and 
$  \atloprat{\Group}( \until { \varphi    } { \psi    }) $
capture the notion
of \emph{rational} strategic capability: 
$  \atloprat{\Group}
\nexttime \varphi$
has to be read 
``coalition $\Group$ has 
a \emph{rational} strategy at its disposal 
which guarantees that
$\varphi$ is going to be true in the next state'',
$  \atloprat{\Group}\henceforth \varphi  $
has to be read 
``coalition $\Group$ has 
a \emph{rational} strategy at its disposal 
which guarantees that   
 $\varphi$ will always be true''. 
Finally, 
$  \atloprat{\Group}( \until { \varphi    } { \psi    }) $
has to be read 
``coalition $\Group$ has 
a \emph{rational} strategy at its disposal 
which guarantees that   
 $\varphi$ will be true until $\psi$ is true''.

Formulas of the language 
$\lang_{\atlratlogic}(\ATM, \AGT )$
are evaluated relative to a pair   
$(P,w)$
with 
$P=(M,\Omega_M)$
a CGSP, 
 $M = ( W,  \ACT,  \allowbreak \relActJoint, 
  \valProp  )$
  a CGS,
  $\Omega_M$
  a preference structure for $M$
  and $w \in W$, as follows:
 \begin{eqnarray*}
   (P,w ) \models   p  & \Longleftrightarrow  &  
  p\in    \valProp\big(\lambda(0) \big)
 ,\\
  (P,w ) \models   \atlop{\Group}
 \nexttime \varphi  & \Longleftrightarrow  &  
  \exists \strategymap_\Group \in  \stratsetatl^\Group_M
  \text{ s.t. }\forall \lambda \in \outset(w, \strategymap_\Group),
  \\ & &
\big( P, \lambda(1)  \big)\models \varphi ,\\
  (P,w ) \models   \atlop{\Group}
 \henceforth  \varphi  & \Longleftrightarrow  &  
  \exists \strategymap_\Group \in  \stratsetatl^\Group_M
  \text{ s.t. }\forall \lambda \in \outset(w, \strategymap_\Group),\\
&& \forall k > 0, \big( P, \lambda(k)  \big)\models \varphi ,\\
(   P,w)\models  \atlop{\Group}( \until { \varphi    } { \psi    })
     & \Longleftrightarrow  &  
   \exists  \strategymap_\Group
  \in \stratsetatl^\Group_M \text{ s.t. }
   \forall  \lambda \in \outset(w, \strategymap_\Group) , \\
   && \exists k > 0 \text{ s.t. } \big(  P,\lambda(k) \big)  \models \psi \text{ and } \\
&& \forall 
   h \in \{ 1, \ldots, k-1\}  , \big( P, \lambda(h) \big) \models \varphi ,\\
     (P,w ) \models   \atloprat{\Group}
 \nexttime \varphi  & \Longleftrightarrow  &  
  \exists \strategymap_\Group \in  \stratsetatl^\Group_M
  \text{ s.t. }
  \forall i \in \Group, 
   \\ & & \strategymap_\Group|_{\{i\}}\not \in \mathit{Dom}_{M,w}^i \text{ and }\\
 && \forall \lambda \in \outset(w, \strategymap_\Group),
\big( P, \lambda(1)  \big)\models \varphi ,\\
  (P,w ) \models   \atloprat{\Group}
 \henceforth  \varphi  & \Longleftrightarrow  &  
  \exists \strategymap_\Group \in  \stratsetatl^\Group_M
  \text{ s.t. }   \forall i \in \Group, 
   \\ & & \strategymap_\Group|_{\{i\}}\not \in 
  \mathit{Dom}_{M,w}^i \text{ and } \\
  && \forall \lambda \in \outset(w, \strategymap_\Group),\\
&& \forall k > 0, \big( P, \lambda(k)  \big)\models \varphi ,\\
(   P,w)\models  \atloprat{\Group}( \until { \varphi    } { \psi    })
     & \Longleftrightarrow  &  
   \exists  \strategymap_\Group
  \in \stratsetatl^\Group_M \text{ s.t. }
    \forall i \in \Group, \\
    &&
    \strategymap_\Group|_{\{i\}}\not \in \mathit{Dom}_{M,w}^i \text{ and }\\
&&   \forall  \lambda \in \outset(w, \strategymap_\Group) , \\
   && \exists k > 0 \text{ s.t. } \big(  P,\lambda(k) \big)  \models \psi \text{ and } \\
&& \forall 
   h \in \{ 1, \ldots, k-1\}  , \big( P, \lambda(h) \big) \models \varphi ,
\end{eqnarray*}
where $ \strategymap_\Group|_{\{i\}}$
is the restriction
of function
$\strategymap_\Group$
to $\{i\}\subseteq \Group$.
Note that the difference between
the strategic capability
operators
and the \emph{rational}
strategic capability
operators
lies in the restriction
to non-dominated (minimally rational)
strategies. 
While the  $\atllogic$
strategic capability 
operators existentially quantify
over
the set of collective strategies 
of the coalition $\Group$ (i.e., $  \exists \strategymap_\Group \in  \stratsetatl^\Group_M$),
their  rational  
counterparts
existentially
quantify 
over
the set of collective strategies 
of the coalition $\Group$
such that all 
their individual
components  are not dominated
(i.e., $\forall i \in \Group, 
    \strategymap_\Group|_{\{i\}}\not \in 
  \mathit{Dom}_{M,w}^i$).

The following fragment 
defines 
the language of $\clratlogic$
($\cllogic$ with \emph{Minimal Rationality}), denoted
by $\lang_{\clratlogic}(\ATM, \AGT )$: 
\begin{center}\begin{tabular}{lcl}
 $\varphi, \psi$  & $\bnf$ & $p \mid \neg\varphi \mid \varphi \wedge \psi  \mid
 \atlop{\Group}
\nexttime \varphi \mid  \atloprat{\Group}
\nexttime \varphi, $
\end{tabular}\end{center}
where  $p$
ranges over  $\ATM$
and $\Group $ ranges over $2^\AGT$. 

The languages $\lang_{\atllogic}(\ATM, \AGT )$ of $\atllogic$
and $\lang_{\cllogic}(\ATM, \AGT )$ of $\cllogic$
are defined as usual:
\begin{itemize}
    \item
    $\lang_{\atllogic}(\ATM, \AGT )$ is the fragment of
$\lang_{\atlratlogic}(\ATM, \AGT )$
with no formulas
 $\atloprat{\Group}
\nexttime \varphi ,
 \atloprat{\Group}
\henceforth  \varphi 
,   \atloprat{\Group}
 (\until{\varphi}{\psi} )$, and
 \item 
  $\lang_{\cllogic}(\ATM, \AGT )$ is the fragment of
$\lang_{\clratlogic}(\ATM, \AGT )$
with no formulas
 $\atloprat{\Group}
\nexttime \varphi$.
\end{itemize}

 \begin{example}[Crossing road (cont.)]
%\textcolor{red}{WIP}  
 
 Returning to our example, it is easy to check that 
 $(P_{croos}, w_0) \models \atloprat{v_1} \nexttime \neg crash  
 $ 
 that is, agent $v_1$ has a rational strategy to avoid a collision. 
However, the agent $v_1$ has no rational strategy to ensure to \textit{eventually} cross the street, that is, 
$(P_{croos}, w_0) \not \models  \atloprat{v_1}  \top \untill c_1 
 $.

 \end{example}

\section{Tree-like  model property and embedding }\label{sec:resultsblock}

In this section we first  state
the tree-like model property for the language $\lang_{\atlratlogic}(\ATM, \AGT )$.
Thanks to it, we will provide 
a polynomial embedding
of the $\atlratlogic$-language
into the 
$\atllogic$-language
which also offers  a polynomial embedding of 
the
 $\clratlogic$-language
into the
$\cllogic$-language.
Thanks to the embedding 
we will be able
to provide tight complexity results
for satisfiability checking for the 
two languages $\lang_{\atlratlogic}(\ATM, \AGT )$
and
$\lang_{\clratlogic}(\ATM, \AGT )$.

\subsection{Tree-like  model property}

Let
 $\relAct{ }^*$, $\relAct{ }^-$ and $\relAct{ }^+$
be, respectively, the reflexive and  transitive closure,
the inverse, and the transitive closure of 
$\relAct{ } = \bigcup_{  \jactatm \in \JACT }  \relAct{  \jactatm}$.
\begin{definition}\label{def:propertiesCGS}
  Let $M = ( W,  \ACT,   \relActJoint,  \valProp  )$
  be a CGS.
  We say that:
  \begin{itemize}
    \item $M$ has a unique root iff
          there is  a unique $ w_0 \in W$ (called the \emph{root}),
          such that, for every $v \in W$, $w_0 \relAct{ }^* v$;
    \item $M$ has unique predecessors iff
          for every $v \neq w_0 $, the cardinality of $\relAct{ }^-( v)$ is at most one;
    \item $M$ has no cycles iff $\relAct{ }^+$ is irreflexive;

    \item $M$ is tree-like iff  it has a unique root,
 unique predecessors and no cycles;

     \item $M$
    is joint action disjoint iff for every $w \in W$
    and for every $\jactatm, \jactatm' \in \JACT $,
    if $\jactatm\neq \jactatm'$
    then $ \relAct{\jactatm} (w) \cap \relAct{\jactatm' } (w) =\emptyset$.
    
  \end{itemize}

\end{definition}

The 
property of 
``having stable preferences''
defined in Definition \ref{def:cgspref}
and the 
properties of ``having unique root'',
``having unique predecessors'',
``having no cycles'',
``being tree-like''
and ``being joint action disjoint''
defined in %the previous
Definition \ref{def:propertiesCGS}
are abbreviated $\mathit{sp},\mathit{ur},\mathit{up} ,
\mathit{nc}, \allowbreak \mathit{tr}$ and $ \mathit{ad} $. 
  The properties defined in %the previous
  Definition \ref{def:propertiesCGS}
  naturally extend to CGSPs: the CGSP
  $P=(M,\Omega_M)$ satisfies one of these properties 
  if the underlying CGS $M$
  satisfies it.  
  For every $X\subseteq \{\mathit{sp}, \mathit{ur},\mathit{up} ,
\mathit{nc}, \mathit{tr}, \allowbreak \mathit{ad}  \}$,
the class of CGS satisfying the properties in $X$
is denoted by $\classcgspvar{X}$. 
%New
By $\classcgspvar{\emptyset}$, we denote the class of all CGSP.

The following Lemma 
\ref{lemmacrucial}
is a tree-like model property for the language $\lang_{\atlratlogic}(\ATM, \AGT )$.
      The proof of the lemma 
     is given in Appendix A in the supplementary material. 
     The proof relies on a three-step transformation. First, we transform
     a CGSP into a CGSP with joint action disjointness by constructing
     one copy of a state for each possible joint action. Second, we transform
     the resulting CGSP with joint action disjointness into a CGSP with joint action disjointness, unique
     predecessor and no cycles. This second transformation associates every state
     of the original CGSP to a finite path. Third, we generate the submodel
     from the point of evaluation of the original model
     to guarantee unique rootness. 
\begin{lemma}\label{lemmacrucial}
  Let $\varphi \in \lang_{\atlratlogic}(\ATM, \AGT )$.
  Then,
  \begin{itemize}
      \item $\varphi$
      is satisfiable 
      for the class $\classcgspvar{\emptyset } $
      iff  $\varphi$
      is satisfiable 
      for the class $\classcgspvar{\{ \mathit{tr},\mathit{ad}  \}}$,

      \item $\varphi$
      is satisfiable 
      for the class $\classcgspvar{\{\mathit{sp}  \}}$
      iff  $\varphi$
      is satisfiable 
      for the class $\classcgspvar{\{ \mathit{sp},\mathit{tr},\mathit{ad}  \}}$. 
  \end{itemize}

\end{lemma}

\subsection{Embedding}

Let us consider 
the following translation
$\mathit{tr}:
\lang_{\atlratlogic}(\ATM, \AGT ) \longrightarrow
\lang_{\atllogic}(\ATM^+, \AGT )
$ 
with $\ATM^+ =\ATM\cup  \big\{ \mathit{rat}_i \suchthat
i \in \AGT 
\big\}$:
\begin{align*}
  \mathit{tr}(p ) &= p,\\
\mathit{tr}(\neg \varphi  ) &= \neg \mathit{tr}( \varphi  ),\\
\mathit{tr}( \varphi \wedge
\psi ) &=   \mathit{tr}( \varphi  ) \wedge \mathit{tr}( \psi   ) ,\\
\mathit{tr}( \atlop{\Group}
\nexttime \varphi ) &=  \atlop{\Group} \nexttime \mathit{tr}( \varphi  )  ,\\
\mathit{tr}( \atlop{\Group}
\henceforth \varphi ) &=  \atlop{\Group}\henceforth \mathit{tr}( \varphi  ),\\
\mathit{tr}\big( \atlop{\Group}
(\until{\varphi}{\psi})  \big) &=  \atlop{\Group} \big(\until{\mathit{tr}(\varphi)}{\mathit{tr}(\psi)}  \big),\\
\mathit{tr}( \atloprat{\Group}
\nexttime \varphi ) &=  \atlop{\Group} \nexttime
\big(\mathit{rat}_C \wedge \mathit{tr}( \varphi  )\big)  ,\\
\mathit{tr}( \atloprat{\Group}
\henceforth \varphi ) &=  \atlop{\Group}\henceforth \big(\mathit{rat}_C \wedge \mathit{tr}( \varphi  )\big),\\
\mathit{tr}\big( \atloprat{\Group}
(\until{\varphi}{\psi})  \big) &=  \atlop{\Group}\big(
\until{
(\mathit{rat}_C \wedge \mathit{tr}(\varphi) )}
{(\mathit{rat}_C \wedge \mathit{tr}(\psi) ) } \big)  ,
\end{align*}
 with $\mathit{rat}_C\eqdef \bigwedge_{i \in \AGT}  \mathit{rat}_i$
 and the special atomic formula 
 $\mathit{rat}_i$
 standing for ``agent $i$
 is rational''.

The idea of the translation
is to transform 
a rational capability
operator
into its 
ordinary capability
counterpart using the special
atomic formulas 
$\mathit{rat}_i$. Specifically, the fact that
a coalition $\Group$
has a rational strategy to ensure a given outcome 
is translated into the fact that 
the coalition $\Group$
has a strategy
to force the  outcome 
by ensuring that all its members are rational.
As the following theorem highlights,
satisfiability
of $\atlratlogic$-formulas
is 
reducible to 
satisfiability
of $\atllogic$-formulas
using 
the translation 
$  \mathit{tr}$.
      The proof of the theorem
     is given in Appendix B in the supplementary material.
     The proof relies on a  non-trivial
     construction which 
     transforms
     a tree-like CGSP into a new tree-like  CGSP
     in which an atomic formula of type $   \mathit{rat}_i$
     matches 
     the computations that are generated by a 
     non-dominated strategy  of agent $i$.
     The assumption of stable preferences is essential
     to guarantee that this matching exists.

 \begin{theorem}\label{theo:embedding}
     Let $\varphi \in \lang_{\atlratlogic}(\ATM, \AGT )$.
     Then, 
     $\varphi $
     is satisfiable for  the class 
     $\classcgspvar{\{\mathit{sp} \}}$
     iff $\big( \bigwedge_{i \in \AGT }
     \atlop{ \{i \}} \henceforth  
      \mathit{rat}_i \big)  \wedge
     \mathit{tr}(\varphi) $
     is satisfiable for the class
          $\classcgspvar{\{\mathit{sp} \}}$.
     % \begin{align*}
     %     \models_{\classcgspvar{\{\mathit{sp} \}}} \varphi \text{ iff }
     %  \big\{ \atlop{ \{i \}} \henceforth  
     %  \mathit{rat}_i \suchthat i \in \AGT 
     %  \big\}       \models_{\classcgspvar{\{\mathit{sp} \}}}  \mathit{tr}(\varphi). 
     % \end{align*}
 \end{theorem}

%   For every $X\subseteq \{\mathit{sp}, \mathit{ur},\mathit{up} ,
% \mathit{nc}, \mathit{tr},\mathit{ad}  \}$,
% the class of CGS satisfying the properties in $X$
% is denoted by $\classcgspvar{X}$. 

 % Deduction theorem:
 % \begin{theorem}
 %     Let $\varphi \in \lang_{\atllogic}(\ATM, \AGT )$
 %     and $\Sigma \subset  \lang_{\atlratlogic}(\ATM, \AGT )$
 %     finite. 
 %     Then, 
 %     \begin{align*}
 %      \Sigma          \models_{\classcgspvar{\{\mathit{sp} \}}}   \varphi 
 %       \text{ iff } 
 %          \models_{\classcgspvar{\{\mathit{sp} \}}}
 %        \atlop{ \emptyset } 
 %                \big( \bigwedge_{\psi \in \Sigma}
 %                \henceforth^* 
 %        \psi  \big) 
 %       \rightarrow \varphi. 
 %     \end{align*}
 % \end{theorem}\

As the following theorem highlights, 
  if we restrict  to the language
$ \lang_{\clratlogic}(\ATM, \AGT )$
the translation $\mathit{tr}$
also  
provides
an  embedding for the general
class
 $\classcgspvar{ }$.
The proof of  Theorem \ref{theo:embedding2}
is a straightforward
adaptation 
of the proof of Theorem \ref{theo:embedding}.
Instead of matching
an atomic formula
  $   \mathit{rat}_i$
  with a computation,
  for every state in a tree-like CGSP
  we match  $   \mathit{rat}_i$ with
  a successor 
  of this state along a computation generated  by a 
     non-dominated strategy  of agent $i$.
 The assumption of stable preferences
 is no longer required
 since the translation
 of formula $ \atloprat{\Group}
\nexttime \varphi $
only refers
to the truth values
of atoms $   \mathit{rat}_i$
in the next state. 

 \begin{theorem}\label{theo:embedding2}
     Let $\varphi \in \lang_{\clratlogic}(\ATM, \AGT )$.
     Then, 
     $\varphi $
     is satisfiable for  the class 
     $\classcgspvar{} $
     iff $\big( \bigwedge_{i \in \AGT }
     \atlop{ \{i \}} \nexttime 
      \mathit{rat}_i \big)  \wedge
     \mathit{tr}(\varphi) $
     is satisfiable for the class
          $\classcgspvar{} $.
     
     % \begin{align*}
     %     \models_{\classcgspvar{\{\mathit{sp} \}}} \varphi \text{ iff }
     %  \big\{ \atlop{ \{i \}} \henceforth  
     %  \mathit{rat}_i \suchthat i \in \AGT 
     %  \big\}       \models_{\classcgspvar{\{\mathit{sp} \}}}  \mathit{tr}(\varphi). 
     % \end{align*}
 \end{theorem}

The following complexity result
is a direct corollary of Theorems \ref{theo:embedding} and \ref{theo:embedding2},
the fact that the size of $\mathit{tr}(\varphi) $
is polynomial
in the size of the input formula $\varphi$
and the fact that satisfiability checking
for $\atllogic$
is \Exptime-complete \cite{DrimmelenLICS}
and
satisfiability checking
for $\cllogic$
is \Pspace-complete \cite{DBLP:journals/logcom/Pauly02}.

  \begin{corollary}
Checking satisfiability of
formulas in the language $ \lang_{\atlratlogic}(\ATM, \AGT )$ relative to  the class    $\classcgspvar{\{\mathit{sp} \}}$ 
is  \Exptime-complete.
It is \Pspace-complete
relative to  
both classes   
$\classcgspvar{ } $
and $\classcgspvar{\{\mathit{sp} \}}$, 
when 
restricting to the fragment $ \lang_{\clratlogic}(\ATM, \AGT )$. 
  \end{corollary}

  Before concluding this section,
  we would like to highlight the fact that the translation
  $\mathit{tr}$
  from
  the language $\lang_{\atlratlogic}(\ATM, \AGT )$ to
  the language $\lang_{\atlratlogic}(\ATM^+, \AGT )$
  is adequate 
for  the stable preference
semantics
only. It does not
work for the general
 class $\classcgspvar{}$. 
% works
% the  model class
% $\classcgspvar{\{\mathit{sp} \}}$
% but does not
% work
% for the general
% class
% $\classcgspvar{}$.
% In other words,
% the translation
%   $\mathit{tr}$
% is adequate 
% for  the stable preference
% semantics
% only. 
To see this,
it is sufficient to observe
that, on the one hand, 
the following formula
is valid
for the general 
class 
$\classcgspvar{}$:
\begin{align*}
  \phi_{\Group,p}  =_{\mathit{def}}\atlop{\Group}\henceforth (\mathit{rat}_\Group \wedge p)\to \atlop{\Group}\nexttime \atlop{\Group}\henceforth (\mathit{rat}_\Group \wedge p). 
\end{align*}
Indeed, 
$  \phi_{\Group,p}$
is a basic
validity
of $\atllogic$. 
Moreover, 
we have 
\begin{align*}
\mathit{tr}\big(
\atloprat{\Group}\henceforth p\to \atlop{\Group}\nexttime \atloprat{\Group}\henceforth p
\big) =
\phi_{ \Group,p}. 
\end{align*}
But, on the other hand,
the formula 
$\atloprat{\Group}\henceforth p\to \atlop{\Group}\nexttime \atloprat{\Group}\henceforth p$
is not valid
for the 
class 
$\classcgspvar{}$, 
which is the same thing
as saying that 
$\neg \big( \atloprat{\Group}\henceforth p\to \atlop{\Group}\nexttime \atloprat{\Group}\henceforth p\big) $
is  satisfiable for $\classcgspvar{}$. 
A counter-model
for this formula
is given in Appendix C. % in the supplementary material. 

 Thus,  there is no analog
 of Theorem \ref{theo:embedding}
 for the class
 $\classcgspvar{}$
 since there exists  a formula
 $\varphi $
 (i.e., $\neg \big( \atloprat{\Group}\henceforth p\to \atlop{\Group}\nexttime \atloprat{\Group}\henceforth p\big) $)
which 
 is  satisfiable for
 $\classcgspvar{}$
 and, at the same time, 
 $
 \big( \bigwedge_{i \in \AGT }
      \atlop{ \{i \}} \henceforth 
       \mathit{rat}_i \big)  \wedge
 \mathit{tr}(\varphi)$
is not satisfiable for 
$\classcgspvar{}$
since $\neg  \mathit{tr}(\varphi)$
(i.e., $ \neg \neg \phi_{\Group,p}$
which is equivalent to $  \phi_{\Group,p}$)
is valid
for 
$\classcgspvar{}$.

     % Let $\varphi \in \lang_{\atlratlogic}(\ATM, \AGT )$.
     % Then, 
     % $\varphi $
     % is satisfiable for  the class 
     % $\classcgspvar{\{\mathit{sp} \}}$
     % iff $\big( \bigwedge_{i \in \AGT }
     % \atlop{ \{i \}} \henceforth  
     %  \mathit{rat}_i \big)  \wedge
     % \mathit{tr}(\varphi) $
     % is satisfiable for the class
     %      $\classcgspvar{\{\mathit{sp} \}}$.

\section{Axiomatization for $\clratlogic$}\label{axiomatization}
In this section, we first introduce an axiomatic system for $\clratlogic$ and then show its soundness and completeness.

\begin{definition}[Axiomatic system for $\clratlogic$]
\label{definition: axiomatic system for R-CL}
The axiomatic system for $\clratlogic$ consists of the following axioms:
\begin{align}
&  \text{All tautologies of propositional logic} \tag{$\top$} \label{ax:taut} 
\\
&  \neg \atlop{\Group}\nexttime \bot \tag{$\mathtt{A}\text{-}\mathtt{NAAA}$} \label{ax:NAAA} 
\\
& \atloprat{\emptyset } \nexttime \phi \leftrightarrow \atlop{\emptyset } \nexttime \phi \tag{$\mathtt{A}\text{-}\mathtt{NP}_{\emptyset}$} \label{ax:NP}
\\
& 
\atloprat{\Group} \nexttime \phi \to \atlop{\Group } \nexttime \phi \tag{$\mathtt{A}\text{-}\mathtt{MR}$} \label{ax:MR}
\\
& 
\atlop{\emptyset} \nexttime (\phi \to \psi) \rightarrow (\atlop{\Group}\nexttime \phi \rightarrow \atlop{\Group} \nexttime \psi ) \tag{$\mathtt{A}\text{-}\mathtt{MG}0$} \label{ax:MG0}
\\
& 
\atloprat{\emptyset} \nexttime (\phi \to \psi) \rightarrow (\atloprat{\Group}\nexttime \phi \rightarrow \atloprat{\Group} \nexttime\psi) \tag{$\mathtt{A}\text{-}\mathtt{MG}1$} \label{ax:MG1}
\\
&
\atlop{\Group}\nexttime \phi \rightarrow \atlop{\Group'}\nexttime \phi, \text{ for }\Group \subseteq \Group' \tag{$\mathtt{A}\text{-}\mathtt{MC}0$} \label{ax:MC0}
\\ 
& 
\atloprat{\Group}\nexttime \phi \rightarrow \atloprat{\Group'}\nexttime \phi , \text{ for }\Group \subseteq \Group' 
\tag{$\mathtt{A}\text{-}\mathtt{MC}1$} \label{ax:MC1}
\\ 
&
\atlop{\Group}\nexttime \top
\tag{$\mathtt{A}\text{-}\mathtt{NCS}$} \label{ax:Ser}
\\ 
&
(\atlop{\Group} \nexttime \phi \land \atlop{\Group'} \nexttime \psi) \rightarrow   \atlop{\Group \cup \Group'} \nexttime (\phi \land \psi),  \nonumber  \\ & 
\text{ for }\Group \cap \Group' = \emptyset
\tag{$\mathtt{A}\text{-}\mathtt{Sup}0$} \label{ax:Sup0}
\\ 
&
(\atloprat{\Group} \nexttime \phi \land \atloprat{\Group'} \nexttime \psi) \rightarrow \atloprat{\Group \cup \Group'} \nexttime (\phi \land \psi),  \nonumber  \\ &  
\text{ for }\Group \cap \Group' = \emptyset
\tag{$\mathtt{A}\text{-}\mathtt{Sup}1$} \label{ax:Sup1}
\\ 
&
\atlop{\Group} \nexttime (\phi \lor \psi) \rightarrow (\atlop{\Group} \nexttime \phi \lor \atlop{\AGT} \nexttime \psi)
\tag{$\mathtt{A}\text{-}\mathtt{Cro}$} \label{ax:Cro}
\\ 
&
\atloprat{\AGT} \nexttime (\phi \lor \psi) \rightarrow (\atloprat{\AGT} \nexttime \phi \lor \atloprat{\AGT} \nexttime \psi)
\tag{$\mathtt{A}\text{-}\mathtt{DGRC}$} \label{ax:DGRA}
\end{align}
and the following rules of inference:
\begin{align}
& \dfrac{\;\phi, \phi \rightarrow \psi \;}
{\;\psi \;} \tag{$\mathtt{MP}$} \label{ir:MP}
\\
& \dfrac{\; \phi \;}{\; \atlop{\emptyset}\nexttime \phi \;} 
  \tag{$\mathtt{N}$} \label{ir:N}
\end{align}
\end{definition}

The names of axioms and inference rules reflect  their intuitions. Axiom \ref{ax:NAAA} is called \textit{no absurd available action} and its intuition is that a coalition's available joint action
 cannot
ensure a logically absurd result. Axiom \ref{ax:NP} is called \textit{no preference for empty coalition}. As the empty coalition has no preference, its rational capability coincides  with its ordinary capability. Axiom \ref{ax:MR} is called \textit{monotonicity of rational capability}:
if an outcome can be ensured by a coalition
in a rational way 
then it can be ordinarily  ensured by the coalition. 
Axiom \ref{ax:MG0} captures \textit{monotonicity of goals}.
Axiom
\ref{ax:MG1} 
is its rational counterpart, namely, 
\textit{monotonicity of goals under rationality}. Their intuition is that ordinary
and rational capabilities are monotonic with respect to goals. Similarly, Axiom \ref{ax:MC0} and \ref{ax:MC1} capture,  respectively,  \textit{monotonicity of coalitions} and \textit{monotonicity of coalitions under rationality}:
ordinary
and rational capability are monotonic with respect to coalitions. Axiom \ref{ax:Ser} captures 
\textit{non-empty choice set},
namely, the fact that a coalition has always
an available joint action. 
Axiom \ref{ax:Sup0} and Axiom \ref{ax:Sup1} capture, respectively, \textit{superadditivity} and \textit{superadditivity under  rationality}.  Axiom \ref{ax:Cro} is the so-called \textit{crown}
axiom: it was called this way  in 
\cite{goranko_strategic_2013}
since it corresponds to the fact that the effectivity function of a game is a crown. Axiom \ref{ax:DGRA}  captures \textit{determinism of 
the grand coalition's rational collective choice}. The inference rules are \textit{modus ponens} (\ref{ir:MP}) and \textit{necessitation for the empty coalition} (\ref{ir:N}).
Note that $ \atlop{\emptyset}\nexttime$
is a normal modal operator because
of the validity-preserving rule
of inference \ref{ir:N}
and the fact that the following formula is valid:
\begin{align*}
    \atlop{\emptyset} \nexttime (\phi \to \psi) \rightarrow (\atlop{\emptyset }\nexttime \phi \rightarrow \atlop{\emptyset} \nexttime \psi ), 
\end{align*}
which is 
  an instance of Axiom 
 \ref{ax:MG0}.

The style of our axiomatic system differs
from $\cllogic$'s \cite{DBLP:journals/logcom/Pauly02}.
We want to get the axiomatization
as close as possible to the semantics
by having as many correspondences as possible between
axioms and semantic constraints. 
In particular,
we have the following correspondences:
Axiom  \ref{ax:Ser} 
corresponds to 
Constraint C3 in Definition \ref{CGS};
Axioms \ref{ax:Sup0} and 
\ref{ax:Sup1}
correspond to Constraint  C2;
Axioms \ref{ax:Cro} and 
\ref{ax:DGRA} correspond to  Constraint C1.

In Appendix D
in the supplementary material 
we show that the
 axiomatic system
 of 
 $\cllogic$ 
 is derivable from $\clratlogic$. 
 Note that the formula 
$\neg \atlop{\emptyset }\nexttime \neg \phi \to \atlop{\AGT}\nexttime \phi $ is an axiom of $\cllogic$. However, by Fact \ref{fact: some invalidities}
whose proof is given in 
Appendix C
in the supplementary material, $\neg \atloprat{\emptyset }\nexttime \neg \phi \to \atloprat{\AGT}\nexttime \phi $ is not valid. So,  the logic for the fragment of $\clratlogic$ only containing rational capability operators  
is substantially
different from  $\cllogic$.
\begin{fact}\label{fact: some invalidities}
The following
two formulas are not valid
for the class  $\classcgspvar{} $: 
  \begin{align}
    & \neg \atloprat{\emptyset }\nexttime \neg \phi \to \atloprat{\AGT}\nexttime \phi \tag{$\mathtt{Max}_{\AGT}1$} \label{invalid: Max1}
    \\
    & \atloprat{\Group} \nexttime (\phi \lor \psi) \rightarrow (\atloprat{\Group} \nexttime \phi \lor \atloprat{\AGT} \nexttime \psi)
    \tag{$\mathtt{Cro}1$} \label{invalid: Cro1}
  \end{align}
\end{fact}

As usual, for every $\varphi \in \lang_{\clratlogic}(\ATM, \AGT )$,
we write $\vdash
\varphi$
to mean that $\varphi $
is deducible  in $\clratlogic$,
that is,
there is a sequence of formulas $(\varphi_1, \ldots, \varphi_m)$
such that:
\begin{itemize}
\item $\varphi_m = \varphi$, and
\item for every $1 \leq k \leq m$, either $\varphi_k$
is  an instance of one of the  axiom schema of   $\clratlogic$
or there are formulas $\varphi_{k_1}, \ldots ,\varphi_{k_t} $
such that $k_1 , \ldots, k_t < k$  and $\frac{\varphi_{k_1}, \ldots ,\varphi_{k_t} }{\varphi_{k} }$
is an instance of some inference  rule of $\clratlogic$.
\end{itemize}

 We are going
 to prove soundness
 and completeness
 of the logic 
 $\clratlogic$ 
 relative to the 
 model class
 $\classcgspvar{ } $.
 So, in the rest of this section,
 when talking about validity
 of a formula 
 $\varphi \in \lang_{\clratlogic}(\ATM, \AGT )$
 we mean validity
 of $\varphi $
 relative to the class  $\classcgspvar{ } $. 
 
Our  completeness proof 
for
$\clratlogic$ 
differs from Pauly's
one
for $\cllogic$ \cite{DBLP:journals/logcom/Pauly02}. It 
is structured in four parts: an induction on the
modal degree of formulas, the normal form (Lemma \ref{lemma: normal-form}), the downward validity (Lemma \ref{lemma: downward validity}), and the upward derivability (Lemma \ref{lemma: upward derivability}). Before introducing  them, we need to introduce  some preliminary notions. 
\begin{definition}[Literal]
A propositional literal
is either $p$
or $\neg p $
 for any  $p \in \ATM $. 
 A modal
 $\clratlogic$-literal
 is a formula of type 
 $\atloprat{\Group}\nexttime \phi$ or $\neg \atloprat{\Group}\nexttime \phi$,
 for any coalition $\Group$ and $\phi\in \lang_{\clratlogic}(\ATM, \AGT )$.
\end{definition}

The following definition introduces
the notion of
standard $\clratlogic$ disjunction.
\begin{definition}[Standard $\clratlogic$ disjunction]
A formula of type 
  \begin{align*}
     \chi \vee & \big(
  (\bigwedge_{\nindex\in \NI}\atlop{\Group_{\nindex}}\nexttime \psi_{\nindex} 
  \wedge 
  \bigwedge_{\nindex\in \NIrat}\atloprat{\Group_{\nindex}}\nexttime \psi_{\nindex})
  \to \\
&  (\bigvee_{\pindex\in \PI}\atlop{\Group_\pindex}\nexttime \psi_{\pindex }
  \vee 
  \bigvee_{\pindex\in \PIrat}\atloprat{\Group_\pindex}\nexttime \psi_\pindex)\big)   
  \end{align*}
is called standard $\clratlogic$ disjunction, where $\NI $, $\NIrat $, $\PI $ and $\PIrat $ are four finite and pairwise disjoint sets of indices, $\chi $ is a disjunction of propositional literals, $\Group_\nindex$ is a coalition and $\psi_\nindex\in \lang_{\clratlogic}(\ATM, \AGT )$ for each $\nindex\in \NI \cup \NIrat \cup \PI \cup \PIrat $.
\end{definition}
The following is a normal form
lemma for the logic $\clratlogic$.
It is proved in a way analogous to the normal form lemma for propositional logic.
\begin{lemma}[Normal form]\label{lemma: normal-form}
  Any  formula
  $\varphi \in \lang_{\clratlogic}(\ATM, \AGT )$
  is equivalent to a conjunction of standard $\clratlogic$ disjunctions whose modal degrees are not higher
  than the modal
  degree of $\varphi$.
  This equivalence is both valid and derivable.
\end{lemma}

Let 
 \begin{align*}
     \chi \vee & \big(
  (\bigwedge_{\nindex\in \NI}\atlop{\Group_{\nindex}}\nexttime \psi_{\nindex} 
  \wedge 
  \bigwedge_{\nindex\in \NIrat}\atloprat{\Group_{\nindex}}\nexttime \psi_{\nindex})
  \to \\
&  (\bigvee_{\pindex\in \PI}\atlop{\Group_\pindex}\nexttime \psi_{\pindex }
  \vee 
  \bigvee_{\pindex\in \PIrat}\atloprat{\Group_\pindex}\nexttime \psi_\pindex)\big)   
  \end{align*}
  be a standard
  $\clratlogic$
  disjunction. 
The following definition
introduces its  sets
of basic indices.  
  They will be used to state the downward validity lemma and the upward derivability lemma.
\begin{definition}[Basic indices]\label{definition: basic indices}
  Define $X_0=\{\nindex\in \NI\mid \Group_\nindex=\emptyset \}$, $Y_0=\{\pindex\in \PI\mid \Group_\pindex=\AGT\}$ and $Y_1=\{\pindex\in \PIrat\mid \Group_\pindex=\AGT\}$. 
\end{definition}
The last definition we need
is that of neat set of indices. 
\begin{definition}[Neatness]\label{definition: neatness}
  For any  $X\subseteq \NI \cup \NIrat $,
  we say $X$ is neat iff for all $\nindex,\nindex'\in X$, if $\nindex\neq \nindex'$, then $\Group_{\nindex}\cap \Group_{\nindex'}= \emptyset$.  
\end{definition}
The following is our downward validity
lemma. 
Its proof is in 
Appendix E
in the supplementary material. 
\begin{lemma}[Downward validity]\label{lemma: downward validity}
   Let $\phi = \chi \vee (
  (\bigwedge_{\nindex\in \NI}\atlop{\Group_\nindex}\nexttime \psi_\nindex 
  \wedge 
  \bigwedge_{\nindex\in \NIrat}\atloprat{\Group_\nindex}\nexttime \psi_\nindex)
  \to 
  (\bigvee_{\pindex\in \PI}\atlop{\Group_\pindex}\nexttime \psi_\pindex 
  \vee 
  \bigvee_{\pindex\in \PIrat}\atloprat{\Group_\pindex}\nexttime \psi_\pindex))$ be a standard 
    $\clratlogic$
    disjunction.
   If  $\phi$ is valid 
  then following 
validity-reduction condition
is satisfied: 
\begin{itemize}
\item $\chi $
is valid, or
\item there is $X\subseteq \NI$ and $X' \subseteq \NIrat $ such that $X\cup X'$ is neat and one of the following
conditions are met:
  \begin{itemize}
  \item there is $\pindex\in \PI$ such that $\bigcup_{\nindex\in X\cup X'} \Group_\nindex \subseteq \Group_\pindex$ and\\
  $ \bigwedge_{\nindex\in X\cup X'} \psi_\nindex\to (\psi_\pindex\vee \bigvee_{\pindex'\in Y_0}\psi_{\pindex'})$ is valid;
  \item there is $\pindex\in \PIrat$ such that $\bigcup_{\nindex\in X'} \Group_\nindex\subseteq \Group_\pindex$ and \\
  $
 \bigwedge_{\nindex\in X_0\cup X'} \psi_\nindex \to (\psi_\pindex\vee \bigvee_{\pindex'\in Y_0}\psi_{\pindex'})$
   is valid;
  \item $\bigwedge_{\nindex\in X_0\cup X'} \psi_\nindex \to \bigvee_{\pindex\in Y_0\cup Y_1}\psi_{\pindex}$ is valid.
  \end{itemize}
\end{itemize}

\end{lemma}

The following is our upward derivability 
lemma. 
Its proof is in 
Appendix  F
in the supplementary material.

\begin{lemma}[Upward derivability]\label{lemma: upward derivability}
   Let $\phi = \chi \vee (
  (\bigwedge_{i\in \NI}\atlop{\Group_\nindex}\nexttime \psi_\nindex 
  \wedge 
  \bigwedge_{i\in \NIrat}\atloprat{\Group_\nindex}\nexttime \psi_\nindex)
  \to 
  (\bigvee_{j\in \PI}\atlop{\Group_\pindex}\nexttime \psi_\pindex 
  \vee 
  \bigvee_{j\in \PIrat}\atloprat{\Group_\pindex}\nexttime \psi_\pindex))$ 
  be a standard     $\clratlogic$
  disjunction.
     If  $\phi$ is valid 
  then following 
derivability-reduction condition is satisfied:

    \begin{itemize}
\item $\vdash \chi $, or
\item there is $X\subseteq \NI$ and $X' \subseteq \NIrat $ such that $X\cup X'$ is neat and one of the following
conditions are met:
  \begin{itemize}
  \item there is $\pindex\in \PI$ such that $\bigcup_{\nindex\in X\cup X'} \Group_\nindex \subseteq \Group_\pindex$ and \\
  $\vdash \bigwedge_{\nindex\in X\cup X'} \psi_\nindex\to (\psi_\pindex\vee \bigvee_{\pindex'\in Y_0}\psi_{\pindex'})$;
  \item there is $\pindex\in \PIrat$ such that $\bigcup_{\nindex\in X'} \Group_\nindex\subseteq \Group_\pindex$ and \\
  $
   \vdash \bigwedge_{\nindex\in X_0\cup X'} \psi_\nindex \to (\psi_\pindex\vee \bigvee_{\pindex'\in Y_0}\psi_{\pindex'})$;
  \item $\vdash \bigwedge_{\nindex\in X_0\cup X'} \psi_\nindex \to \bigvee_{\pindex\in Y_0\cup Y_1}\psi_{\pindex}$.
  \end{itemize}
\end{itemize}
\end{lemma}

 The following is the culminating
 result of this section:
 the logic $\clratlogic$
 is sound and complete for the model class
 $\classcgspvar{ } $.

 We show Theorem \ref{thm:soundnessCL} by  induction on modal degrees of formulas. The inductive hypothesis ensures that the validity-reduction condition of a formula implies its derivability-reduction condition.  The complete proof is given in Appendix G. 
\begin{theorem}
\label{thm:soundnessCL}
Let 
 $\phi \in \lang_{\clratlogic}(\ATM, \AGT )$.
 Then, $\phi$
 is valid if and only if  $\vdash \phi $. 
\end{theorem}

\section{Model checking}\label{sec:modelchecking}

The global model checking problem for $\atlratlogic$ consists of computing, for a given CGSP $P$, and a formula $\varphi$,  all the states in which $\varphi$ holds in $P$, formally $\{w \in W : (P,w) \models \varphi\}$.   
In this section, we consider this problem relative
to a   subclass of CGSP
in which preferences are short-sighted. 
 \begin{definition}[Short-sighted preferences]\label{def:shortpref}
 Let $(M,\Omega_M)$
be  a CGSP with  $M = ( W,  \ACT,  \allowbreak \relActJoint, 
  \valProp  )$
  a CGS.   We say that $M$
 has short-sighted preferences
          if
          the following condition holds:
          \begin{align*}
       (\mathbf{SSP}) &  \
\forall i \in \AGT, 
\forall w \in W, 
 \forall \lambda, \lambda' \in 
          \historyset_{M,w }
          \text{ if }
          \lambda(1)= \lambda'(1)
       \\
&    \text{ then } \lambda'   \approx _{i,w } \lambda , 
          \end{align*}
          where 
$ \lambda'   \approx _{i,w } \lambda$
iff 
$ \lambda'   \preceq_{i,w } \lambda$
and 
$ \lambda   \preceq_{i,w } \lambda'$.
\end{definition}
The short-sighted preference
condition means that an agent 
is indifferent between computations that
are equal until the next state. 
Notice that if a CGSP $M$ has
both stable and short-sighted preferences
in the sense of Definitions
\ref{def:cgspref}
and 
\ref{def:shortpref}
then the following holds:
\begin{align*}
\forall i \in \AGT, 
\forall v \in W,     \text{ if }
   v \in \mathcal{R}(w_0)
    \text{ then }
    \forall \lambda, \lambda' \in
              \historyset_{M,v },       
    \lambda'   \approx _{i,v } \lambda . 
\end{align*}
This means that under stable and 
short-sighted preferences only 
the successor states of 
the initial state $w_0$
affect an agent's preferences 
 since from the next state on an agent has complete indifference between computations.

The reason why we verify  properties relative to CGSPs
with short-sighted preferences
is to have an efficient model-checking procedure. Indeed, verifying properties
with respect to the general class of CGSPs 
would make model checking exponential since we would need to
compute dominance by alternating between 
sets of strategies of exponential size (similar to \cite{AminofGR21}). 

The proof of Theorem \ref{thm:modelchecking} is given in Appendix H.  The lower-bound follows from the model checking of  \atllogic %which is \Ptime-complete 
\cite{alur2002alternating}. 
For the upper bound, we first define agents' preference relation in state $w$ over the successors of $w$. This allows us to define the notion of agents' dominated actions  at a given state $w$. We then reinterpret the rational strategic modalities $\atloprat{\Group}\nexttime$, $\atloprat{\Group}\henceforth$, and $\atloprat{\Group}\until{}{}$ over dominated actions instead of dominated strategies, which is equivalent for the case of GCSP with short-sighted preferences. Then, we extend the model-checking algorithm for \atllogic to include  the modalities  $\atloprat{\Group}\nexttime$, $\atloprat{\Group}\henceforth$, and $\atloprat{\Group}\until{}{}$. The resulting algorithm, provided in Appendix H, runs in polynomial time.  

\begin{theorem}
\label{thm:modelchecking}
    The global model checking problem for \atlratlogic over GCSP with short-sighted preferences is \Ptime-complete. 
\end{theorem}
%\begin{proof}[Proof Sketch] \end{proof}

%$Q \subseteq S$, the set of states, 
%\mnote{\textcolor{blue}{PS: We cannot remove the strategies from function $Dom()$, because preferences are defined over pairs of computations. An action profile leads to infinite computations, and it is not clear how to compare them without associating them with a strategy (can we, instead,  include finite paths in the preference? that would solve the issue). The function $Dom$ does not use temporal reasoning, as preferences are binary relations. Elimination of strongly dominated strategies with binary preferences was shown to be polynomial in the literature, although it grows on the size of the possible strategy profiles.  We need to restrict to a finite set of strategies, thus the memoryless assumption. }} 

\balance

\section{Conclusion}
\label{sec:conclusion}

We have proposed a novel
semantic 
analysis of preferences
in concurrent games and used
our semantics based
on CGS with preferences
to define a new family
of $\atllogic$
and $\cllogic$
languages distinguishing
the notion
of ordinary
capability
from
the notion
of \emph{rational}
capability.
We have provided
a variety
of proof-theoretic
and complexity
results for our languages
with an emphasis on both
satisfiability
checking and model checking. 

Directions of future work are manifold.
Some proof-theoretic
aspects remain to be explored and complexity results
to be proved. Future work 
will be devoted to i) axiomatizing the full logic
$\atlratlogic$ relative to class
$\classcgspvar{} $ and to class 
 $\classcgspvar{\{\mathit{sp} \}}$, ii) studying complexity
of satisfiability
checking
for the language $\lang_{\atlratlogic}(\ATM, \AGT )$ relative to the general
class    $\classcgspvar{} $.
In our running example, the joint strategy in which both agents cross the road  is not dominated. This is because preferences are defined for individual agents rather than coalitions. 
Following previous work on group preference logic \cite{DBLP:journals/jancl/Lorini11} and judgment aggregation \cite{grossi2014judgment,baumeister2017strategic}, we plan to extend our framework with the notion
of group preferences resulting from the aggregation of individual preferences. 
We also intend to 
investigate levels of rationality, where agents assume minimal rationality of their opponents, allowing for iterated strong dominance  \cite{Bonanno2008,Lorini13}. 
Last but not least, we plan to consider an epistemic extension 
of our semantics and languages 
to be able to model concurrent games with imperfect information
and an agent's knowledge of its rational capability. 
\ifarxiv
Since model checking \atllogic with imperfect information and perfect recall is undecidable \cite{abs-1102-4225}, we may consider models with public actions \cite{BelardinelliLMR20} and hierarchical information \cite{BerthonMMRV21}.
\fi

\begin{acks} 
%\ifarxiv
Support from 
the ANR projects EpiRL 
``Epistemic Reinforcement Learning''
(grant number ANR-22-CE23-0029)
and  ALoRS
``Action, Logical Reasoning and Spiking networks''
(grant number ANR-21-CE23-0018-01)
is    acknowledged. 
This project has received funding from the European Union’s Horizon 2020 research and innovation programme under the Marie Skłodowska-Curie grant agreement No 101105549. 
%\else
%Shorter version: 
%This research has received funding from the  EU H2020 Marie Sklodowska-Curie project with grant agreement No 101105549.
%\fi
\end{acks}

%\newpage 

\bibliographystyle{ACM-Reference-Format}
%\bibliography{biblio}

%To print the appendix uncomment:
\ifarxiv
\newpage \input{appendix}
\fi
 
\end{document}

%% file: figModel.tex
\begin{figure}[h]

   \centering

     \scalebox{.7}{
     \begin{tikzpicture}[shorten >= 1pt, shorten <= 1pt,  auto]
  \tikzstyle{rond}=[circle,draw=black,minimum size  = 1.25cm] 
   \tikzstyle{label}=[sloped,shift={(0,-.1cm)}]
  
  \node[rond,fill=white] (q0) { \begin{tabular}{c}   \end{tabular}};
  %\node[below=0.1cm of q0]{$init$};
  \node[rond,fill=white] (q4) [right= 2cm of q0] {\begin{tabular}{c} $crash$  \end{tabular}};
  \node[rond,fill=white] (q2) [ below= 2cm   of q4] {\begin{tabular}{c} $c_2$  \end{tabular}};
   \node[rond,fill=white] (q1) [ above= 2cm  of q4] {\begin{tabular}{c} $c_1$   \end{tabular}}; 
  \node[rond,fill=white,mystyle] (q3) [ right= 2cm  of q4] {\begin{tabular}{c} $\begin{tabular}{c} $c_1,c_2$\end{tabular}%w_3
  $  \end{tabular}};

  \node  [below= 0.1cm  of q0] {$q_0$};
  \node  [below= 0.1cm  of q1] {$q_1$};
  \node  [below= 0.1cm  of q2] {$q_2$};
  \node  [below= 0.1cm  of q3] {$q_3$};
  \node  [below= 0.1cm  of q4] {$q_4$};

%bend left= 20, \begin{tabular}{c} $(t,t)/1$ \\  $(t,d)/0.5 $  \end{tabular}
 \path[->,thick] (q0) edge  node [label,pos=.5]    {$(Skip,Move)$}    (q2);
 \path[->,thick] (q2) edge  node [label,pos=.5]    {\begin{tabular}{c} $(Move,Skip)$\\$(Move,Move) $\end{tabular}
 }    (q3);
 \path[->,thick] (q0) edge  node [label,pos=.5]    {$(Move,Skip)$}    (q1);
  \path[->,thick] (q1) edge  node [label,pos=.5]    {\begin{tabular}{c} $(Skip,Move)$\\$(Move,Move) $\end{tabular}
  }    (q3);
 \path[->,thick] (q0) edge  node [label,pos=.5]    {$(Move,Move)$}    (q4);

\path [->,thick] (q0) edge[loop left]  node  [pos=.3]   {$(Skip,Skip)$} (q0);
\path [->,thick] (q2) edge[loop left]  node  [pos=.3]   {\begin{tabular}{c} $(Skip,Skip)$\\$(Skip,Move) $\end{tabular}}  (q2);
\path [->,thick] (q1) edge[loop left]  node  [pos=.3]   {\begin{tabular}{c} $(Skip,Skip)$\\$(Move,Skip) $\end{tabular}}  (q1);
\path [->,thick] (q3) edge[loop right]  node  [pos=.3]   {%\begin{tabular}{c} $(Skip,Skip)$ \\  $(Skip,Move)$\\$(Move,Skip)$ \end{tabular}
$(*,*)$
}  (q3);
\path [->,thick] (q4) edge[loop right]  node  [pos=.3]   {$(*,*)$}  (q4);

\end{tikzpicture}
}

     \caption{Model $M_{cross}$ representing a  system with two vehicles approaching an intersection. 
     Arrows represent transitions between states and are labeled by joint actions of $v_1$ and $v_2$. $(*,*)$ denotes any action.
     } 
     \label{fig:model} 
 \end{figure}
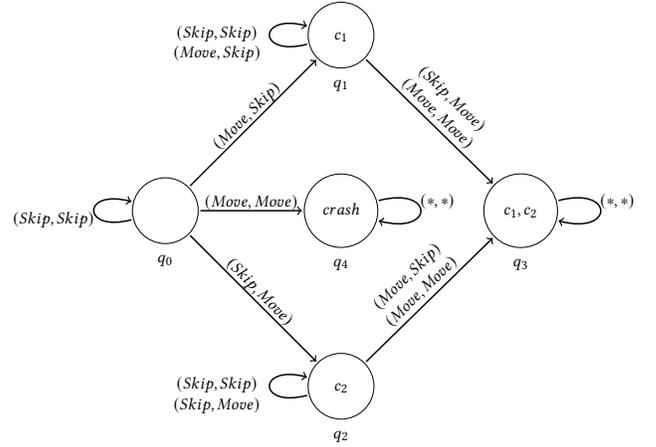

%% file: figPref.tex
\begin{figure}[h]

   \centering

     \scalebox{.7}{
     \begin{tikzpicture}[shorten >= 1pt, shorten <= 1pt,  auto]
  \tikzstyle{rond}=[circle,draw=black,minimum size  = 1.25cm] 
   \tikzstyle{label}=[sloped,shift={(0,-.1cm)}]
  
  \node[rond,fill=white] (q0) {};

  %2nd level
  \node[rond,fill=white] (q1) [below left= 2.5cm and 1cm of q0] { $c_2$};
  \node[rond,fill=white] (q2)  [left= 1cm of q1] {$c_1$};
  \node[rond,mystyle,fill=black!20] (q3)  [right= 1cm of q1] {$crash$};
  \node[rond,fill=black!20] (q4) [right= 1cm of q3]  { };
  \path[->,thick] (q0) edge  node [label,pos=.5]    {$(Skip, Move)$}    (q1);
  \path[->,thick] (q0) edge  node [label,pos=.5]    {$(Move, Skip)$}    (q2);
  \path[->,thick] (q0) edge  node [label,pos=.5]    {$(Move, Move)$}    (q3);
  \path[->,thick] (q0) edge  node [label,pos=.5]    {$(Skip, Skip)$}    (q4);
  \path [->,thick] (q3) edge[loop left,looseness=4]  node  [pos=.5,xshift=.6cm, yshift=0.5cm
  ]   {$(*,*)$}  (q3);
  %\path [->,dotted,thick,bend right] (q4) edge  node [label,pos=.5]    {}  (q0);
  \path [->,%dotted,
  thick] (q0) edge[loop left,looseness=4]  node [label,pos=.5]    {}  (q0);

  %3rd level
  \node[rond,fill=black!20] (q6)  [below= 2cm of q2] {$c_1$};
  \node[rond,mystyle,fill=black!20] (q5)  [left= 1cm of q6]   {$c_1,c_2$};
  \path[->,thick] (q2) edge  node [label,pos=.5]    {$(*,*)$}    (q6);
  \path[->,thick] (q2) edge  node [label,pos=.5]    {$(*,*)$}    (q5);
  \path [->,thick] (q5) edge[loop left,looseness=4]  node  [pos=.5,xshift=.6cm, yshift=0.5cm]   {$(*,*)$}  (q5); 
 % \path [->,dotted,thick,bend right] (q6)
  edge  node [label,pos=.5]    {}  (q2);
  \path [->,%dotted,
  thick] (q2) edge[loop left,looseness=4]  node [label,pos=.5]    {}  (q2);

  \node[rond,mystyle,fill=black!20] (q7)  [below= 2cm of q1] {$c_1,c_2$};
  \node[rond,mystyle,fill=black!20] (q8)  [right= 1cm of q7]   {$c_2$};
  \path[->,thick] (q1) edge  node [label,pos=.5]    {$(Move, *)$}    (q7);
  \path[->,thick] (q1) edge  node [label,pos=.5]    {$(Skip, *)$}    (q8);
  \path [->,thick] (q7) edge[loop left,looseness=4]  node  [pos=.5,xshift=.6cm, yshift=0.5cm]   {$(*,*)$}  (q7);
  %\path [->,dotted,thick,bend right] (q8) 
   edge  node [label,pos=.5]    {}  (q1);
   \path [->,%dotted,
   thick] (q1) edge[loop left,looseness=4]  node [label,pos=.5]    {}  (q1);

  \node  [right= 0.01cm  of q0] {$q_0$};
  \node  [right= 0.01cm  of q4] {$q_0$};

  \node  [right= 0.01cm  of q2] {$q_1$};
  \node  [right= 0.01cm  of q6] {$q_1$};

  \node  [right= 0.01cm  of q1] {$q_2$};
  \node  [right= 0.01cm  of q8] {$q_2$};
  
  \node  [right= 0.01cm  of q5] {$q_3$};
  \node  [right= 0.01cm  of q7] {$q_3$};

  \node  [right= 0.01cm  of q3] {$q_4$};

  \node[below=0.1cm of q5]{$(+_1,+_2)$}; %c1,c1c2 %++
  \node[below=0.1cm of q6]{$(+_1,=_2)$}; %loop c1
  \node[below=0.1cm of q7]{$(+_1,+_2)$};%c2,c1c2
  \node[below=0.1cm of q8]{$(=_1,+_2)$};%c2,
  \node[below=0.1cm of q4]{$(=_1,=_2)$};%init,
  \node[below=0.1cm of q3]{$(-_1,-_2)$};%init,

  %\node[rond,fill=white,mystyle] (q) [ right= 2cm  of q4] {$c_1,c_2$};

%bend left= 20, \begin{tabular}{c} $(t,t)/1$ \\  $(t,d)/0.5 $  \end{tabular}
% \path[->,thick] (q0) edge  node [label,pos=.5]    {}    (q1);
 
%\path [->,thick] (q4) edge[loop right]  node  [pos=.3]   {}  (q4);

\end{tikzpicture}
}

     \caption{Representation of the unravelling of $M_{cross}$ from the initial state ($w_0$). Branches represent (groups of) computations. Transitions are labeled by the action taken by $v_1$ and $*$ denotes any action. Self-loops indicate computations where the state is repeated. %Dotted arrows indicate computations that repeat in a given state but may proceed to a different state.
     %preference relations $\preceq_{v_1,w_0}$ and $\preceq_{v_2,w_0}$ 
     Grey states indicate computations with an infinite suffix that repeats on the same state. Labels in the form under the grey states represent the preference relations  $\preceq_{v_1,w_0}$ and $\preceq_{v_2,w_0}$  of the agents $v_1$ and $v_2$, respectively.  Computations labeled with $+_i$ are strictly  preferred to $=_i$ by agent $v_i$, and  $=_i$ are strictly preferred to~$-_i$ by agent $v_i$ (where $i = \{1,2\}$).     } 
     \label{fig:pref} 
 \end{figure}
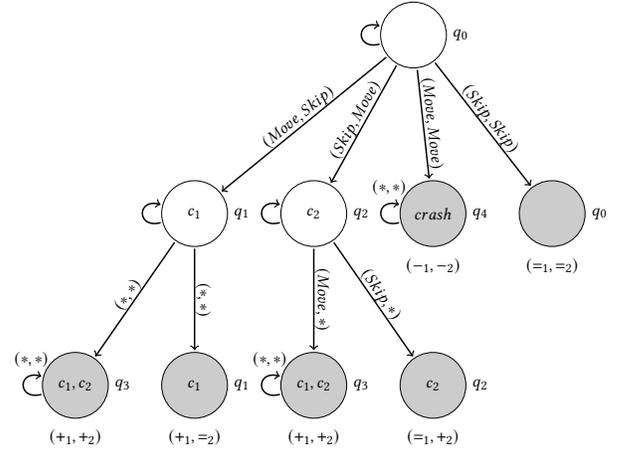

%% file: appendix.tex
\appendix

\section{Proof of Lemma 
\ref{lemmacrucial}}\label{prooflemmatree}

\begin{proof}
    The right-to-left direction of the lemma is evident.
    We prove the left-to-right direction
    for the first item.  
The proof for the second  item
is analogous. 

The proof is organized in three steps.

\textbf{Step 1}. 
    Let 
    $(P,w_0 )\models \varphi $
with 
$P=(M,\Omega_M) \in
\classcgspvar{ }$, 
 $M = ( W,  \ACT, \\ (\relAct{\jactatm}   )_{\jactatm \in \JACT}, 
  \valProp  )$,
  $\Omega_{M }=(\preceq_{i,w }   )_{i \in \AGT, w \in W }$
  and $w_0 \in W $
  such that $(P,w_0)\models \varphi$. 

  We are going to transform 
  $P=(M,\Omega_M)$ into a new CGSP 
    $P'=(M',\Omega_{M'})$
    with 
     $M' = ( W',  \ACT,  (\relAct{\jactatm}'   )_{\jactatm \in \JACT}, 
  \valProp'   )$
  and $\Omega_{M'}=(\preceq_{i,\sigma  }'   )_{i \in \AGT, \sigma  \in W' }$. 
  The construction goes as follows: 

     \begin{itemize}
\item $W' = \{w^\jactatm \suchthat
w \in W \text{ and }
\jactatm \in \JACT
\}$, 

\item for all $w^\jactatm , v^{\jactatm' }  \in W'$
and for all $\jactatm'' \in \JACT$,
$w^\jactatm  \relAct{\jactatm''}'  v^{\jactatm' }$
iff $w  \relAct{\jactatm''}  v$
and 
$\jactatm''=\jactatm' $,

\item for all $p \in \ATM$ and for all $w^\jactatm  \in W'$,
$p \in \valProp'( w^\jactatm )$
iff $p \in \valProp( w )$,

\item 
for all $w^\jactatm  \in W'$,
for all $i \in \AGT $, for all $w^\jactatm  v^{\jactatm_1}
 u^{\jactatm_2} \ldots , \allowbreak
w^\jactatm  t^{\jactatm_3}
 z^{\jactatm_4} \ldots \in \historyset_{M',w^\jactatm }$, 
we have 
$$w^\jactatm  v^{\jactatm_1}
 u^{\jactatm_2} \ldots \preceq_{i,w^\jactatm   }' 
 w^\jactatm  t^{\jactatm_3}
 z^{\jactatm_4} \ldots \text{ iff }$$
%iff 
$$ w v
 u \ldots \preceq_{i,w   }  w  t
 z \ldots   $$

\end{itemize}

  It is routine to verify that $P'=(M',\Omega_{M'})
\in \classcgspvar{\{  \mathit{ad}  \}} $. 
and, by induction on the structure of $\varphi$,
that
$(P,v )  \models \varphi$
iff
$(P',v^\jactatm  )  \models \varphi$
for all $v^\jactatm \in W' $
and for all $v \in W $. 
Hence, $(P', w_0^{\jactatm} ) \models \varphi$
for every $\jactatm \in \JACT $
since 
$(P, w_0 ) \models \varphi$. 

  For notational convenience,
  in the rest of the proof 
  we 
 denote with 
 $x_1, x_2, \ldots $
 the
  elements
  of $W'$
in the CGSP $M'$
defined above.

\textbf{Step 2}. 
% % The second  step follows the proof of 
% % \cite[Lemma 5.4]{DBLP:conf/atal/BoudouL18}.
% We define the set of tracks in $M'$, 
% denoted by $\trackset_{M' }$,
% a track
% being a non-empty finite sequence
% $x_0  \jactatm_1 x_1 \ldots\jactatm_{k} x_k  $
% such that (i) $x_0 \in W'$,
% (ii) $ \jactatm_1 x_1 \ldots\jactatm_{k} x_k  \in 
% (\JACT \times W')^k$ for some $k \in \mathbb{N }_0$, and
% (iii)
% for every
% $0 \leq h \leq k-1$,
% $x_{h }    \relAct{\jactatm_{h +1 }} x_{h +1 }$.
% Elements of $\trackset_{M'}$ are denoted by $\sigma,\sigma', \ldots$
% For every $\sigma \in \trackset_{M'}$,
% we denote by $\sigma[\mathit{last}]$
% the last element in the sequence $\sigma $.
We are going to transform 
  $P'=(M',\Omega_{M'})$ into a new CGSP 
    $P''=(M'',\Omega_{M''})$
    with 
     $M'' = ( W'',  \ACT,  (\relAct{\jactatm}''   )_{\jactatm \in \JACT}, 
  \valProp''   )$
  and $\Omega_{M''}=(\preceq_{i,\sigma  }''   )_{i \in \AGT, \sigma  \in W'' }$. 
  Before defining $P''$
  and $M''$
  we introduce some further notation.
We denote the set of finite paths
in $M'$
by  $ \finpathset_{M'}$
and its elements by $\pi, \pi', \ldots$.
The last element in a finite path $\pi$
is denoted by $\pi[\mathit{last}]$.
  
  The definition of
  $P''$
  and $M''$
  goes as follows:
     \begin{itemize}
\item $W'' = \finpathset_{M'}$,

\item for all $\pi= x_1\ldots x_k ,
\pi' = x_1'\ldots x_h'
\in W''$
and for all $\jactatm \in \JACT$,\\
$x_1\ldots x_k
\relAct{\jactatm}''  x_1'\ldots x_h'$
iff $h=k+1$,
$x_1=x_1', \ldots, x_k=x_k'$
and $x_k \relAct{\jactatm}' x_h' $,

\item for all $p \in \ATM$ and for all $\pi  \in W''$,
$p \in \valProp''( \pi )$
iff $p \in \valProp'( \pi [\mathit{last}])$,

\item 
for all $\pi  \in W''$,
for all $i \in \AGT $, for all $\pi \pi_1
\pi_2 \ldots ,
\pi \pi_1'
\pi_2' \ldots \in \historyset_{M'',\pi  }$, 
we have 
$\pi \pi_1
\pi_2 \ldots \preceq_{i,\pi   }''  \pi \pi_1'
\pi_2' \ldots $
iff 
$$
\pi[\mathit{last}]
\pi_1[\mathit{last}]
\pi_2[\mathit{last}]\ldots 
\preceq_{i,\pi [\mathit{last}]  }'  \pi [\mathit{last}]
\pi_1'[\mathit{last}]
\pi_2' [\mathit{last}]\ldots  $$ %.

\end{itemize}

It is routine to verify that $P''=(M'',\Omega_{M''})
\in \classcgspvar{\{  \mathit{up} ,
\mathit{nc},\mathit{ad}  \}} $.

We  define the function $f$
mapping worlds in $M''$
into worlds in $M' $.
Let $x \in W' $
and $\pi  \in W'' $.
Then,
$f(\pi   ) = x $
iff
$x=\pi  [\mathit{last}]$. 
% \redtext{
%   We  define the function $g$
% mapping strategies in $M'$
% into strategies in $M $.
% Let $\strategymap_\Group \in \stratsetatl_{M}^\Group $
% and 
%  $\strategymap'_\Group \in \stratsetatl_{M'}^\Group $ for all coaltion $\Group $.
%  Then, 
%  $g(\strategymap'_\Group) = \strategymap_\Group $
% iff
% $\strategymap_\Group(a)(f(\sigma_0)\ldots f(\sigma_k))=\strategymap'_\Group(a)(\sigma_0\ldots \sigma_k)$ for all $\sigma_0\ldots \sigma_k\in \pathset_{M' }$ and $a\in \Group$.
% }

% \redtext{
  By induction on the structure of $\varphi$,
  it is straightforward 
  to show that
$(P',x )  \models \varphi$
iff
$(P'',\pi   )  \models \varphi$
for all $\pi  \in W''$
and for all $x \in W' $
such that $f(\pi) =x $.
% The Boolean cases 
%   are straightforward. Compared to the case $\atloprat{\Group}
% (\until{\psi_1}{\psi_2 })$, the cases $\atlop{\Group}
% \nexttime \psi $,
% $\atlop{\Group}
% \henceforth \psi $, $\atlop{\Group}
% (\until{\psi_1}{\psi_2 })$, $\atloprat{\Group}
% \nexttime \psi $, and
% $\atloprat{\Group}
% \henceforth \psi $
%    are evident. So we just show the case $\atloprat{\Group}
% (\until{\psi_1}{\psi_2 })$.
Hence, $(P'', w_0^{\jactatm}  ) \models \varphi$
for all 
 $\jactatm \in \JACT $
since $w_0^{\jactatm}  \in W'$,
$w_0^{\jactatm}  \in W''$,
$f(w_0^{\jactatm} )=w_0^{\jactatm} $
and 
$(P', w_0^{\jactatm}  ) \models \varphi$
for all 
 $\jactatm \in \JACT $.

\textbf{Step 3}. 
We fix a joint action 
 $\jactatm_0 \in \JACT $
 and we  transform the CGSP
  $P=(M'',\Omega_{M'' })$
  defined in the previous step 
  into a new structure 
    $P'''=(M''',\Omega_{M'''})$
    with 
     $M''' = ( W''',  \ACT,  (\relAct{\jactatm}'''  )_{\jactatm \in \JACT}, \allowbreak
  \valProp'''  )$
  and $\Omega_{M'''}=(\preceq_{i,\sigma  }'''   )_{i \in \AGT, \sigma \in W''' }$. 
  The construction goes as follows:
  \begin{itemize}
  \item $W''' = \{ \pi  \in W'' \suchthat w_0^{\jactatm_0}  \relAct{ }^* \pi  \} $,
  \item $\relAct{\jactatm}''' = \relAct{\jactatm}''
  \cap( W'' \times W'') $ for all $\jactatm \in \JACT$,
  \item $ \valProp''' (\pi  ) = \valProp'' (\pi  ) $ for all $\pi  \in W'''$,

  \item $\preceq_{i,\pi   }''' =\preceq_{i,\pi   } '' $ for all $\pi   \in W'''$
  and for all $i \in \AGT $.
  \end{itemize}
    $P'''=(M''',\Omega_{M'''})$
   is nothing but 
  the submodel of     $P''=(M'',\Omega_{M'' })$
  generated by the state $w_0^{\jactatm_0}  $. 
  Clearly, $ P'''=(M''',\Omega_{M''' })\in 
\classcgspvar{\{ \mathit{tr},\mathit{ad}  \}} $
since $P''=(M'',\Omega_{M''})
\in \classcgspvar{\{  \mathit{up} ,
\mathit{nc},\mathit{ad}  \}} $. Moreover,
by induction
on the structure of $\varphi$,
it is routine to verify
that 
``$(P'', \pi   )\models \varphi $ iff
  $(P''' ,\pi   )\models \varphi $'' for all $\pi \in W'''$. 
  Thus, $(P''' ,w_0^{\jactatm_0} )\models \varphi $ since
  $(P '',w_0^{\jactatm_0}  )\models \varphi $.
 \end{proof}

\section{Proof of Theorem \ref{theo:embedding}}
\label{annex:theoremembedding}

\begin{proof} 
We first prove the left-to-right direction. Suppose $\varphi$
is satisfiable for the class $\classcgspvar{\{\mathit{sp} \}}$.
Thus, by Lemma
\ref{lemmacrucial},
there exists 
$P=(M,\Omega_{M}) \in
\classcgspvar{\{ \mathit{sp},\mathit{tr},\mathit{ad}  \}} $, 
 $M = ( W,  \ACT,  (\relAct{\jactatm}   )_{\jactatm \in \JACT}, 
  \valProp  )$,
  $\Omega_{M }=(\preceq_{i,w }   )_{i \in \AGT, w \in W }$
  and $w_0 \in W $
  such that $(P,w_0 )\models \varphi $. 
  Next, we will define a new structure
  $P'=(M',\Omega_{M})   $
with
 $M' = ( W,  \ACT,  \allowbreak (\relAct{\jactatm}  )_{\jactatm \in \JACT}, \allowbreak
  \valProp ' )$
  such that  for all $ v\in W$,
  we have 
  \begin{align*}
\valProp '(v)= &
   \big(\valProp (v)\cap \ATM\big) 
   \cup \\
   & \Big\{ \mathit{rat}_i
   \suchthat 
        \exists  u \in \relAct{ }^-( v),
     \exists \strategymap_{\{i\}} \in  \big( \stratsetatl^{\{i\}}_{M} 
     \setminus \mathit{Dom}_{M,u}^i
     \big),
     \\ &
     \exists 
        \lambda 
    \in \outset(u,      \strategymap_{\{i\}} )
     \text{ s.t. }
    v=  \lambda(1)
   \Big\}.      
  \end{align*}
Clearly, 
    $P'=(M',\Omega_{M}) \in
\classcgspvar{\{ \mathit{sp},\mathit{tr},\mathit{ad}  \}} $
since 
$P=(M,\Omega_{M}) \in
\classcgspvar{\{ \mathit{sp},\mathit{tr},\mathit{ad}  \}} $
and the only difference between $P'$
and $P$
is in the valuation function 
that  does not affect the relational properties
of the structure. 
By induction on the structure of the formula,
we are going to prove that 
``$(P',v )\models \mathit{tr}( \varphi )$ iff
  $(P,v )\models \varphi $'' for all $v \in W$.
  The Boolean cases 
  are straightforward 
  and the cases $  \atlop{\Group}
\nexttime \psi $,
$\atlop{\Group}
\henceforth \psi $
and $\atlop{\Group}
(\until{\psi_1}{\psi_2 })$
   are evident. 
Let us prove the case 
$\varphi= \atloprat{\Group}
\nexttime \psi  $.
We have 
$(P,v )\models \atloprat{\Group}
\nexttime \psi $ 
iff 
$\exists \strategymap_\Group \in  \stratsetatl^\Group_{M}
  \text{ s.t. }
  \forall i \in \Group, \strategymap_\Group|_{\{i\}}\not \in \mathit{Dom}_{M ,v }^i \text{ and } \forall \lambda \in \outset(v, \strategymap_\Group), \allowbreak 
\big( P , \lambda(1)  \big)\models \psi  $.
By construction of $P'$ and the fact that
$P$
and $P'$
satisfy  properties $\mathit{tr}$ and $\mathit{ad }$,
the latter is equivalent to $\exists \strategymap_\Group \in  \stratsetatl^\Group_{M' }
 $ such that $\forall \lambda \in \outset(v, \strategymap_\Group), \allowbreak
\big( P , \lambda(1)  \big)\models 
\bigwedge_{i \in \AGT }\mathit{rat}_i \wedge \psi $. 
By induction hypothesis,
the latter is equivalent to 
$\exists \strategymap_\Group \in  \stratsetatl^\Group_{M'}
  \text{ s.t. }\forall \lambda \in \outset(v, \strategymap_\Group),
\big( P' , \lambda(1)  \big)\models 
\bigwedge_{i \in \AGT }\mathit{rat}_i \wedge \mathit{tr}(\psi)  $. 
The latter is equivalent to
$(P',v )\models \atlop{\Group}
\nexttime \big(\mathit{rat}_C \allowbreak \wedge \mathit{tr}(\psi)  \big)  $.

Let us prove the case 
$\varphi= \atloprat{\Group}
\henceforth  \psi  $.
We have 
$(P ,v )\models \atloprat{\Group}
\henceforth \psi $ 
iff 
$\exists \strategymap_\Group \in  \stratsetatl^\Group_{M }
  \text{ s.t. }   \forall i \in \Group, \strategymap_\Group|_{\{i\}}\not \in 
  \mathit{Dom}_{M,v}^i \text{ and }
  \forall \lambda \in \outset(v, \strategymap_\Group), \allowbreak
 \forall k > 0, \big( P, \lambda(k)  \big)\models \psi   $.
By construction of $P'$ and the fact that
$P$
and $P'$
satisfy  properties $\mathit{tr}$ and $\mathit{ad }$,
the latter is equivalent to $\exists \strategymap_\Group \in  \stratsetatl^\Group_{M' } 
  \text{ s.t. }   
  \forall \lambda \in \outset(v, \strategymap_\Group),
 \forall k > 0, \big( P, \lambda(k)  \big)\models
 \bigwedge_{i \in \AGT }\mathit{rat}_i \wedge 
 \psi  $. 
By induction hypothesis,
the latter is equivalent to 
$\exists \strategymap_\Group \in  \stratsetatl^\Group_{M' }$ 
  such that 
$  \forall \lambda \in \outset(v, \strategymap_\Group),
 \forall k > 0, \big( P' , \lambda(k)  \big)\models
 \bigwedge_{i \in \AGT }\mathit{rat}_i \wedge 
\mathit{tr}(\psi)  $. 
The latter is equivalent to
$(P',v )\models \atlop{\Group}
\henceforth \big(\mathit{rat}_C \wedge \mathit{tr}(\psi)  \big)  $.
We omit the proof for the case 
$\varphi= \atloprat{\Group}
(\until{\psi_1}{\psi_2 } ) $
since it is analogous to the previous case. 

Thus, we can conclude that $(P',w_0 )\models \mathit{tr}( \varphi )$
since 
$(P,w_0 )\models  \varphi  $.
It  is also straightforward to prove that,
      for all  $i\in \AGT$
      we have 
$(P',w_0 )\models \atlop{ \{i \}} \henceforth  
      \mathit{rat}_i$
      since
      i)      
      $  \big( \stratsetatl^{\{i\}}_{M} 
     \setminus \mathit{Dom}_{M,w_0}^i
     \big) \neq \emptyset $, 
     and ii) 
     if 
         $ \strategymap_{\{i\}} \in  \big( \stratsetatl^{\{i\}}_{M} 
     \setminus \mathit{Dom}_{M,w_0 }^i
     \big) \neq \emptyset $
     and $w_0 \relAct{ }^* v$
     then   $ \strategymap_{\{i\}} \in  \big( \stratsetatl^{\{i\}}_{M} 
     \setminus \mathit{Dom}_{M,v}^i
     \big) \neq \emptyset $
     (because of stable preferences for $P$). 

%    \mathit{tr}( \atloprat{\Group}
% \nexttime \varphi ) &=  \atlop{\Group} \nexttime
% \big(\mathit{rat}_C \wedge \mathit{tr}( \varphi  )\big)  ,\\

Let us now prove the right-to-left direction.
Suppose the formula $\big( \bigwedge_{i \in \AGT }
     \atlop{ \{i \}} \henceforth  
      \mathit{rat}_i \big)  \wedge
     \mathit{tr}(\varphi) $
     is satisfiable for the class
          $\classcgspvar{\{\mathit{sp} \}}$. 
          Thus, 
by Lemma 
\ref{lemmacrucial},  there exists  $P=(M,\Omega_M) \in
\classcgspvar{\{ \mathit{sp},\mathit{tr},\mathit{ad}  \}} $,
 $M = ( W,  \ACT,  (\relAct{\jactatm}   )_{\jactatm \in \JACT}, 
  \valProp  )$, $\Omega_{M }=(\preceq_{i,v }   )_{i \in \AGT, v \in W }$
  with  unique root $w_0\in W$ of $M$
such that   $(P,w_0 )\models \atlop{ \{i \}} \henceforth  
      \mathit{rat}_i$
      for all $i\in \AGT$  and 
$(P,w_0 )\models \mathit{tr}(\varphi) $. 

 We are going to define a new structure
  $P'=(M,\Omega_{M}' )   $
with
  $\Omega_{M' }=(\preceq_{i,v } '  )_{i \in \AGT, v \in W }$
  such that  for all $\lambda, \lambda' \in  \historyset_{M,w_0 } $
  and for all $i \in \AGT$, 
  \begin{align*}
\lambda'  \preceq_{i,w_0  }'  \lambda
\text{ iff }
\mathit{U}_{i,w_0 }(\lambda' ) \leq \mathit{U}_{i,w_0 }(\lambda),
  \end{align*}
  where 
    \begin{align*}
& \mathit{U}_{i,w_0 }(\lambda)= 
\begin{cases}
1 & \text{if }   \exists \strategymap_{\{i\}} \in  \stratsetatl^{\{i\}}_M
 \text{ s.t. } \lambda \in \outset(w_0, \strategymap_{\{i\}}  ) \text{ and }\\
&  \forall \lambda' \in  \outset(w_0, \strategymap_{\{i\}}  ),
 \forall k> 0, \big(P', \lambda(k)\big) \models  \mathit{rat}_i , \\
0 &  \text{otherwise};
\end{cases}
      \end{align*}
      and for all $v \in W$
      such that $w_0 \relAct{ }^+  v$
      and for all $\lambda, \lambda' \in  \historyset_{M, v } $, 
  \begin{align*}
\lambda'  \preceq_{i,v   }'   \lambda
\text{ iff }
\mathit{prec}(v)\lambda'  \preceq_{i,\mathit{prec}(v)   }'  \mathit{prec}(v) \lambda, 
  \end{align*}
  where 
  $\mathit{prec}(v)$
  is 
  the unique element of 
  the set 
  $\relAct{ }^-(v)$.

By induction on the structure of the formula,
it is routine to prove that 
``$(P,v )\models \mathit{tr}( \varphi )$ iff
  $(P',v )\models \varphi $'' for all $v \in W$.

\end{proof}

\section{A counter-model}\label{section: a counter-model}
  Consider the case $\AGT = \{a\}$, $M = ( W,  \ACT,  (\relAct{\jactatm})_{\jactatm \in \JACT}, \valProp)$ and $P = ( M, (\preceq_{a,w})_{w \in W })$, where:
  \begin{itemize}
    \item $W=\{w_0,w_1\}$;
    \item $\ACT=\{0,1\}$;
    \item $\relAct{\{(a,0)\}}=\{(w_1,w_0)\}$ and $\relAct{\{(a,1)\}}=\{(w_0,w_1),(w_1,w_1)\}$;
    \item $\Va{w_0}=\emptyset $ and $\Va{w_1}=\{p\} $;
    \item for all $w\in W$ and $\comp, \comp' \in \historyset_{M,w}$: $\comp \preceq_{a,w} \comp' $ iff $\mathit{U}_{a,w}(\comp) \leq \mathit{U}_{a,w}(\comp')$, where 
    \begin{align*}
& \mathit{U}_{a,w_0 }(\comp)= 
\begin{cases}
1 & \text{if }   \forall k\in \mathbb{N}_+: \comp(k)=w_1\\
0 &  \text{otherwise};
\end{cases}\\
& \mathit{U}_{a,w_1 }(\comp)= 
\begin{cases}
  1 & \text{if }  \comp(1)=w_0\\
  0 &  \text{otherwise}.
\end{cases}
      \end{align*}
  \end{itemize}
  \begin{fact}
    \begin{align*}
    & P,w_0\nvDash \atloprat{\AGT }\henceforth p\to \atlop{\AGT }\nexttime\atloprat{\AGT }\henceforth p\\
    & P,w_1\nvDash \neg \atloprat{\emptyset }\nexttime \neg p \to \atloprat{\AGT}\nexttime p \\
    & P,w_1\nvDash \atloprat{\emptyset } \nexttime (\neg p \lor p) \rightarrow (\atloprat{\emptyset } \nexttime \neg p \lor \atloprat{\AGT} \nexttime p) 
    \end{align*}
  \end{fact}
\begin{proof}
  Consider two types of strategies of grand coalition: $\Sigma_0=\{\strategymap\in \stratsetatl_{M}^{\AGT}\mid \strategymap(a)(w_1)=0\}$ and $\Sigma_1=\{\strategymap\in \stratsetatl_{M}^{\AGT}\mid \strategymap(a)(w_1)=1\}$. Note $\comp^{P,w_1,\strategymap_0}(1)=w_0$ for all $\strategymap_0\in \Sigma_0$ and $\comp^{P,w_1,\strategymap_1}(1)=w_1$ for all $\strategymap_1\in \Sigma_1$. Then $\mathit{Dom}_{M,w_1}^a=\Sigma_1$.
  
  It is clear that $(P,w_0)\vDash \atloprat{\AGT }\henceforth p$. Note $P,w_0\vDash \neg p$. Then $(P,w_1)\vDash \neg \atloprat{\AGT }\henceforth p$ and $(P,w_0)\vDash \neg \atlop{\AGT }\nexttime\atloprat{\AGT }\henceforth p$. Then $(P,w_0)\nvDash \atloprat{\AGT }\henceforth p\to \atlop{\AGT }\nexttime\atloprat{\AGT }\henceforth p$.

It is clear that $P,w_1\vDash \neg \atloprat{\emptyset }\nexttime \neg p$ and $P,w_1\vDash \atloprat{\emptyset }\nexttime (\neg p\vee p)$. 
Note $P,w_0\vDash \neg p$. Then $P,w_1\vDash \neg \atloprat{\AGT }\nexttime p$. Therefore, $P,w_1\nvDash \neg \atloprat{\emptyset }\nexttime \neg p \to \atloprat{\AGT}\nexttime p $ and $P,w_1\nvDash \atloprat{\emptyset } \nexttime (\neg p \lor p) \rightarrow (\atloprat{\emptyset } \nexttime \neg p \lor \atloprat{\AGT} \nexttime p) $.
\end{proof}

\section{Some derivabilities}\label{section: some derivabilities}
\begin{theorem}
  The following are derivable in $\clratlogic$:
  %\tag{$\mathtt{A}\text{-}\mathtt{MC}1$} \label{ax:MC1}
  \begin{align}  
    %No absurd rational actions
    & 
    \neg \atloprat{\Group}\nexttime \bot \tag{$\mathtt{NARA}$} \label{NARA}
    \\
    %Seriality with rationality 
    & 
    \atloprat{\Group}\nexttime \top \tag{$\mathtt{Ser}1$} \label{Ser1}
    \\
    %Necessitation with rationality 
    &
    \dfrac{\; \phi \;}
    {\; \atloprat{\emptyset }\nexttime \phi \;}
    \tag{$\mathtt{N}1$} \label{N1}
    \\
    %Monotonicity 
    & \dfrac{\;\phi\to\psi \;}{\;\atlop{\Group}\nexttime \phi \to \atlop{\Group'}\nexttime \psi \;},\text{ where }\Group\subseteq \Group' \tag{$\mathtt{Mon}0$} \label{Mon0}
    \\
    %Monotonicity in rationality
    & 
    \dfrac{\;\phi\to\psi \;}{\;\atloprat{\Group}\nexttime \phi \to \atloprat{\Group'}\nexttime \psi\;},\text{ where }\Group\subseteq \Group'\tag{$\mathtt{Mon}1$} \label{Mon1}
    \\
    %Regularity 
    &
    \atlop{\Group }\nexttime \phi \to \neg \atlop{\overline{\Group }}\nexttime \neg \phi \tag{$\mathtt{Reg}0$} \label{Reg0} 
    \\
    %Regularity with rationality 
    &
    \atloprat{\Group }\nexttime \phi \to \neg \atloprat{\overline{\Group }}\nexttime \neg \phi \tag{$\mathtt{Reg}1$} \label{Reg1} 
    \\
    %Maximality of grand coalition 
    &
    \neg \atlop{\emptyset }\nexttime \neg \phi \to \atlop{\AGT}\nexttime \phi 
    \tag{$\mathtt{Max}_{\AGT}0$} \label{Max0} 
    \\ 
    %Crown in mix 
    &
    \atloprat{\Group} \nexttime (\phi \lor \psi) \rightarrow (\atloprat{\Group} \nexttime \phi \lor \atlop{\AGT} \nexttime \psi) 
    \tag{$\mathtt{Cro}0.5$} \label{Cro0.5}
  \end{align}
  \end{theorem}  
\begin{proof}
  \begin{itemlist}{(\ref{NARA})}
      \item[(\ref{NARA})] From \ref{ax:NAAA} and \ref{ax:MR}.
      
      \item[(\ref{Ser1})] By \ref{ax:Ser}, $\vdash_{\clratlogic}\atlop{\emptyset }\nexttime \top$. By \ref{ax:NP}, $\vdash_{\clratlogic}\atloprat{\emptyset }\nexttime \top$. By \ref{ax:MC1}, $\vdash_{\clratlogic}\atloprat{\Group}\nexttime \top $.
      
      \item[(\ref{N1})] From \ref{ir:N} and \ref{ax:NP}.  
      
      \item[(\ref{Mon0})] Assume $\vdash_{\clratlogic} \phi \to \psi$. By \ref{ir:N}, $\vdash_{\clratlogic} \atlop{\emptyset }\nexttime (\phi \to \psi )$.   
      By \ref{ax:MG0}, we have that $\vdash_{\clratlogic} \atlop{\Group }\nexttime \phi \to \atlop{\Group }\nexttime \psi $. By \ref{ax:MC0}, $\vdash_{\clratlogic} \atlop{\Group }\nexttime \phi \to \atlop{\Group' }\nexttime \psi $
      
      \item[(\ref{Mon1})] Similar to the proof for (\ref{Mon0}).
      
      \item[(\ref{Reg0})] By \ref{ax:Sup0}, $\vdash_{\clratlogic}(\atlop{\Group }\nexttime \phi \wedge \atlop{\overline{\Group }}\nexttime \neg \phi ) \to \atlop{\AGT}\nexttime (\phi \wedge \neg \phi )$. Note $\vdash_{\clratlogic} \atlop{\AGT}\nexttime (\phi \wedge \neg \phi )\to \atlop{\AGT}\nexttime \bot $ and $\vdash_{\clratlogic} \neg \atlop{\AGT}\nexttime \bot $. Then $\vdash_{\clratlogic}\atlop{\Group }\nexttime \phi \to \neg \atlop{\overline{\Group }}\nexttime \neg \phi $
      
      \item[(\ref{Reg1})] From \ref{Reg0} and \ref{ax:MR}.
      
      \item[(\ref{Max0})] By \ref{ir:N}, $\vdash_{\clratlogic} \atlop{\emptyset }\nexttime (\neg \phi\vee \phi)$. By \ref{ax:Cro}, $\vdash_{\clratlogic} \atlop{\emptyset }\nexttime \neg \phi \vee \atlop{\AGT }\nexttime \phi $. Then $\vdash_{\clratlogic}\neg \atlop{\emptyset }\nexttime \neg \phi \to \atlop{\AGT}\nexttime \phi $.
      
      \item[(\ref{Cro0.5})] Note $\vdash_{\clratlogic} ((\phi \lor \psi)\wedge  \neg \psi)\to \phi $. By \ref{ax:Sup1} and \ref{Mon1}, $\vdash_{\clratlogic} (\atloprat{\Group} \nexttime (\phi \lor \psi)\wedge \atloprat{\emptyset }\nexttime \neg \psi) \to \atloprat{\Group }\nexttime \phi$. Then $\vdash_{\clratlogic} \atloprat{\Group} \nexttime (\phi \lor \psi)\to ( \atloprat{\Group }\nexttime \phi \vee \neg \atloprat{\emptyset }\nexttime \neg \psi)$. By \ref{ax:NP} and \ref{Max0}, $\vdash_{\clratlogic} \neg \atloprat{\emptyset }\nexttime \neg \psi \to \atlop{\AGT }\nexttime \psi $. Therefore,  we have 
      $\vdash_{\clratlogic} \atloprat{\Group} \nexttime (\phi \lor  \psi) \rightarrow (\atloprat{\Group} \nexttime \phi \allowbreak \lor \allowbreak \atlop{\AGT} \nexttime \psi)$. 
    \end{itemlist}
\end{proof}

%\section{Downward validity}
\section{Proof of Lemma \ref{lemma: downward validity}}
\label{section: downward validity}
To prove the downward validity lemma, the key is  to construct a model for any satisfiable standard conjunction, which is equal to the negation of a standard disjunction. We use a method called `blueprint' to achieve it.

\subsection{Blueprints and their realizations}
\begin{definition}[Blueprint]

A blueprint is a tuple $\bp = (\ACT_0, \\
\rlist, \clist)$, where:
  \begin{itemize}
  \item $\ACT_0$ is a nonempty set of actions;
  \item $\clist: \JACT^{\AGT}_0 \to \lang_{\clratlogic}(\ATM, \AGT )$ is a consequence-list function, where $\JACT^{\AGT}_0=\{\sigma: \AGT\to \ACT_0\}$, and $\clist(\sigma)$ is satisfiable for all $\sigma \in \JACT_0$.
  \item $\rlist: \AGT \to (\powerset{\ACT}-\{\emptyset \})$ is a rationality-list function;
  \end{itemize}
\end{definition}
\begin{theorem}[Realization]\label{theorem: realization}
  Let $\bp = (\ACT_0 , \clist , \rlist )$ be a blueprint and $\chi$ be a satisfiable conjunction of propositional literals. Then there is a pointed model $(P,w_0)$, called a realization of $\bp$ and $\chi$, where $P=( M , (\preceq_{a,w})_{a\in \AGT, w \in W })$, $M=( W, \ACT, \relActJoint ,  \valProp )$, and $\relActJoint = (\relAct{\jactatm} )_{\jactatm \in \JACT}$ such that:
  \begin{itemize}
    \item $P,w_0\vDash \chi$;
    \item $\ACT_0\subseteq \ACT$;
    \item for all $\strategymap\in \stratsetatl_{M}^\AGT$ and $\sigma\in\JACT^{\AGT}_0$, if $\strategymap(a)(w_0)=\sigma(a)$ for all $a\in \AGT$, then $M,\lambda^{M,w_0,\strategymap}(1)\vDash \clist(\sigma)$;
    \item for all coalition $\Group$, $\strategymap_{\Group}\in \stratsetatl_{M}^{\Group}$ and $a\in \Group$: $\strategymap_{\Group}|_{\{a\}}\not \in \mathit{Dom}_{M,w_0}^a$ iff $\strategymap_{\Group}(a)(w_0)\in \rlist(a)$.
  \end{itemize}
\end{theorem}
\begin{proof}  
  For all $\sigma\in\JACT^{\AGT}_0$, let $(P^{\sigma }, w^{\sigma })$ be a pointed model satisfies $\clist(\sigma )$, where $P=(W^{\sigma }, \ACT^{\sigma }, (\rel^{\sigma}_{\jactatm} )_{\jactatm \in \JACT^{\sigma }},    \allowbreak \valProp^{\sigma }, (\preceq^{\sigma }_{a,w})_{a\in \AGT, w \in W^{\sigma } })$, their domains are pairwise disjoint, and their set of actions are pairwise disjoint and disjoint with $\ACT_0$. Let $w_0$ be a world not in any of these models.
  
  Define the CGS with preferences $P=( W, \ACT, (\relAct{\jactatm} )_{\jactatm \in \JACT},  \allowbreak \valProp , (\preceq_{a,w})_{a\in \AGT, w \in W })$ as the following:
  \begin{itemize}
    \item $W=\{w_0\}\cup \bigcup_{\sigma\in\JACT^{\AGT}_0}W^{\sigma }$;
    \item $\ACT=\ACT_0\cup \bigcup_{\sigma\in\JACT^{\AGT}_0}\ACT^{\sigma }$;
    \item for all $\jactatm: \AGT\to \ACT $:
    $$\rel_{\jactatm}=\begin{cases}
      \rel^{\sigma}_{\jactatm}&\text{ if there is }\sigma\in\JACT^{\AGT}_0 \text{ such that }\\
      &\jactatm: \AGT\to \ACT^{\sigma } \\
      \{(w_0,w^\jactatm)\}&\text{ if }\jactatm\in\JACT^{\AGT}_0 \\
      \emptyset &\text{ otherwise }
    \end{cases}$$
    \item for all $w \in W$:
    $$\Va{w}=$$
    $$\begin{cases}
      \valProp^{\sigma }(w)&\text{ if }w\in W^{\sigma } \\ & \text{ for some }\sigma\in\JACT^{\AGT}_0,\\
      \{p\in \AP \mid p \text{ is a conjunct of }\chi \}&\text{ otherwise }
    \end{cases}$$
    \item for all $w \in W$ and $a\in\AGT$:
    $$\preceq_{a,w}=\begin{cases}
        \preceq^{\sigma }_{a,w}&\text{ if }w\in W^{\sigma } \text{ for some }\sigma\in\JACT^{\AGT}_0\\
        \preceq^0_a&\text{ if }w=w_0
    \end{cases}$$
    where $\preceq^0_a=\{(\lambda,\lambda')\in \historyset_{P,w_0}\times \historyset_{P,w_0} \mid\text{either }\lambda'(1)=w^\sigma \text{ and }\sigma(a)\in \rlist(a), \text{ or }\lambda(1)=w^\sigma \text{ and }\sigma(a)\notin \rlist(a)\\ \text{for some }\sigma\in\JACT^{\AGT}_0\}$. 
  \end{itemize}

  It is clear that $P,w_0\vDash \chi$ and $\ACT_0\subseteq \ACT$.
  
  Let $\strategymap\in \stratsetatl_{M}^\AGT$ and $\sigma: \AGT \to \ACT_0$ such that $\strategymap(a)(w_0)=\sigma(a)$ for all $a\in \AGT$. We want to show $M,\lambda^{M,w_0,\strategymap}(1)\vDash \clist(\sigma)$. Note $\rel_\sigma(w_0)=\{w^\sigma\}$. Then $\lambda^{M,w_0,\strategymap}(1)=w^\sigma$. Therefore, $$M,\lambda^{M,w_0,\strategymap}(1)\vDash \clist(\sigma).$$

  Let $\Group\subseteq \AGT$, $\strategymap_{\Group}\in \stratsetatl_{M}^{\Group}$ and $a\in \Group$. We want to show $\strategymap_{\Group}|_{\{a\}}\not \in \mathit{Dom}_{M,w_0}^a$ iff $\strategymap_{\Group}(a)(w_0)\in \rlist(a)$. 

  Assume $\strategymap_{\Group}(a)(w_0)\in \rlist(a)$. Note that $$\strategymap_{\Group}|_{\{a\}}\oplus \strategymap_{\AGT-\{a\}}(a)(w_0)\in \rlist(a)$$ for all $\strategymap_{\AGT-\{a\}}\in \stratsetatl_{M}^{\AGT-\{a\}}$. Then  $\lambda^{P,w_0,\strategymap_{\Group}|_{\{a\}}\oplus \strategymap_{\AGT-\{a\}}}\succeq \lambda$ for all $\strategymap_{\AGT-\{a\}}\in \stratsetatl_{M}^{\AGT-\{a\}}$ and $\lambda \in \historyset_{P,w_0}$. It follows that for all $\strategymap'_{\{a\}}\in \stratsetatl_{M}^{\{a\}}$, there is $\strategymap_{\AGT-\{a\}}\in \stratsetatl_{M}^{\AGT-\{a\}}$ such that $\lambda^{P,w_0,\strategymap_{\Group}|_{\{a\}}\oplus \strategymap_{\AGT-\{a\}}}\succeq \lambda^{P,w_0,\strategymap'_{\{a\}}\oplus \strategymap_{\AGT-\{a\}}}$. Therefore, $\strategymap_{\Group}|_{\{a\}}\not \in \mathit{Dom}_{M,w_0}^a$.
  
  Assume $\strategymap_{\Group}(a)(w_0)\notin \rlist(a)$. Note $\rlist(a)\neq \emptyset $. Then there is $\strategymap'_{\{a\}}\in \stratsetatl_{M}^{\{a\}}$ such that $\strategymap'_{\{a\}}(a)(w_0)\in \rlist(a)$. Let $\strategymap_{\AGT-\{a\}}\in \stratsetatl_{M}^{\AGT-\{a\}}$. Then $$\lambda^{P,w_0,\strategymap_{\Group}|_{\{a\}}\oplus \strategymap_{\AGT/\{a\}}}\succ \lambda^{P,w_0,\strategymap'_{\{a\}}\oplus \strategymap_{\AGT/\{a\}}}.$$%. 
  Therefore, $\strategymap_{\Group}|_{\{a\}}\in \mathit{Dom}_{M,w_0}^a$.    
\end{proof}
\subsection{The proof}
Now we prove the downward validity lemma (Lemma \ref{lemma: downward validity}).
\begin{proof}
  We prove its contrapositive. Assume validity-reduction-condition is not met. Then it meets satisfiability-reduction-condi\-tion, that is:   
  $\neg \chi $ is satisfiable and for all $X\subseteq \NI$ and $X' \subseteq \NIrat $, if $X\cup X'$ is neat, then:
  \begin{itemize}
    \item for all $\pindex\in \PI$, if $\bigcup_{\nindex\in X\cup X'} \Group_\nindex 
      \subseteq \Group_\pindex$, then $\bigwedge_{\nindex\in X\cup X'} \psi_\nindex 
      \wedge \neg \psi_\pindex\wedge \bigwedge_{\pindex'\in Y_0}\neg \psi_{\pindex'}$ is satisfiable for all $\pindex\in \PI$;
    \item for all $\pindex\in \PIrat$, if $\bigcup_{\nindex\in X'}        \Group_\nindex\subseteq \Group_\pindex$, then $\bigwedge_{\nindex\in X_0\cup X'} \psi_\nindex 
    \wedge \neg \psi_\pindex\wedge \bigwedge_{\pindex'\in Y_0}\neg \psi_{\pindex'}$ is satisfiable;
    \item $\bigwedge_{\nindex\in X_0\cup X'} \psi_\nindex \wedge \bigwedge_{\pindex\in Y_0\cup Y_1}\neg \psi_\pindex$ is satisfiable.
  \end{itemize}

  We want to show $\neg \chi \wedge \bigwedge_{\nindex\in \NI}\atlop{\Group_\nindex}\nexttime \psi_\nindex 
  \wedge 
  \bigwedge_{\nindex\in \NIrat}\atloprat{\Group_\nindex}\nexttime \psi_\nindex\wedge 
  \bigwedge_{\pindex\in \PI}\neg \atlop{\Group_\pindex}\nexttime \psi_\pindex 
  \wedge 
  \bigwedge_{\pindex\in \PIrat}\neg \atloprat{\Group_\pindex}\nexttime \psi_\pindex$ is satisfiable. Note $\phi $ is satisfiable iff $\phi \wedge \atloprat{\emptyset }\nexttime \top \wedge \neg \atloprat{\AGT}\nexttime \bot$ is satisfiable for all formula $\phi $. Therefore, for convenience, we can assume $\NIrat\neq \emptyset $ and $\PI\cup \PIrat$ is a nonempty beginning segment of natural numbers.

  Let $\ACT_0=\NI \cup \NIrat \cup \PI \cup \PIrat $ and $\JACT^{\Group}_0=\{\sigma_{\Group}: \Group \to \ACT_0\}$.
  Define $\Fsupport: \bigcup_{\Group\subseteq \AGT }\JACT^{\Group}_0 \to \powerset{\NI}$, 
  such that for all $\Group \subseteq \AGT$ and $\ja{\Group}\in \JACT^{\Group}_0$:
  \begin{align*}
  \support{\ja{\Group}} = \{x \in \NI\cup \NIrat \mid & \Group_\nindex\subseteq \Group \text{ and } \move{\Group}{a} = \nindex \\ &\text{for all } a \in \Group_\nindex\}    
  \end{align*}
  
  Define $\Fimpeach: \bigcup_{\Group\subseteq \AGT }\JACT^{\Group}_0 \to \PI\cup\PIrat$, such that for all $\Group \subseteq \AGT$ and $\ja{\Group}\in \JACT^{\Group}_0$:
  $$
  \impeach{\ja{\Group}} = \Big( \sum_{a \in \Group} \move{\Group}{a} \Big) \bmod n$$
  
  \noindent where $n$ is the cardinality of $\PI\cup\PIrat$.

  Define $\bp = (\ACT_0, \clist, \rlist)$ as follows:
  \begin{itemize}
    \item $\rlist: \AGT \to \ACT_0$ such that for all $a\in \AGT $:
    $$\rlist(a)=\NIrat $$
    \item $\clist: \JACT^{\AGT}_0 \to \lang_{\clratlogic}(\ATM, \AGT )$ such that for all $\sigma\in \JACT^{\AGT}_0$:
    $$\clist(\sigma)=$$
    $$
    \begin{cases}    
      \theta \wedge \bigwedge_{\pindex\in Y_1}\neg \psi_\pindex&\text{if }{\ja{}(a)}\in \rlist(a)\text{ for all }a\in \AGT;\\
      \theta \wedge \neg \psi_{\impeach{\ja{}}}&\text{if }{\ja{}(a)}\notin \rlist(a)\text{ for some }a\in \AGT,\\
      &\bigcup_{\nindex\in \support{\ja{}}}\Group_\nindex\subseteq \Group_\impeach{\ja{}},\\
      &\text{ and }\impeach{\ja{}}\in \PI;\\
      \theta \wedge \neg \psi_{\impeach{\ja{}}}&\text{if }{\ja{}(a)}\notin \rlist(a)\text{ for some }a\in \AGT,\\
      &\bigcup_{\nindex\in \support{\ja{}}}\Group_\nindex\subseteq \Group_\impeach{\ja{}},\\
      &{\ja{}(a)\in \rlist(a)}\text{ for all }a\in \Group_\impeach{\ja{}},\\
      &\text{ and }\impeach{\ja{}}\in \PIrat;\\
      \theta &\text{otherwise. }
    \end{cases}$$
    where $\theta = \bigwedge_{\nindex\in \support{\ja{}}}\psi_\nindex \wedge \bigwedge_{\pindex \in Y_{0}}\neg \psi_\pindex$.
  \end{itemize}
  
  Now we show $\bp $ is a blueprint. It suffices to show $\clist(\sigma)$ is satisfiable for all $\sigma\in \JACT^{\AGT}_0$ and $\rlist(a)\neq \emptyset$ for all $a\in \AGT$. Note $\NIrat \neq \emptyset $. Then $\rlist(a)\neq \emptyset$ for all $a\in \AGT$.

  Let $\ja{}\in \JACT^{\AGT}_0$. We want to show that $\clist(\sigma)$ is satisfiable. 
  
  Firstly, we will show $\support{\ja{}}$ is neat. 
  
  Let $\nindex,\nindex'\in \support{\ja{}}$ such that $\nindex\neq \nindex'$. Assume $\Group_\nindex\cap \Group_{\nindex'}\neq \emptyset $. Let $a\in \Group_\nindex\cap \Group_{\nindex'}$. By the definition of $\Fsupport$, $\move{}{a}=\nindex$ and $\move{}{a}=\nindex'$. But $\nindex\neq \nindex'$. Then $\Group_\nindex\cap \Group_{\nindex'}= \emptyset $. Therefore, $\support{\ja{}}$ is neat.

  Secondly, we will show that 
  if $\move{}{a}\in \rlist(a)$ for all $a\in \bigcup_{\nindex\in \support{\ja{}}}\Group_\nindex$, then $\support{\ja{}}\subseteq X_0\cup X'$ for some $X'\subseteq \NIrat$.
 
  Assume $\move{}{a}\in \rlist(a)$ for all $a\in \bigcup_{\nindex\in \support{\ja{}}}\Group_\nindex$. We want to show $(\NI\cap \support{\ja{}})\subseteq X_0$. Assume not. Then there is $\nindex\in \NI\cap \support{\ja{}}$ such that $\nindex\notin X_0$. Then $\Group_\nindex\neq \emptyset $. Therefore, there is $a\in \Group_\nindex$ such that $\move{}{a}=i$. Note $\nindex\in \NI$. Then $\nindex\notin \rlist(a)$. It is contradictory to $\move{}{a}\in \rlist(a)$ for all $a\in \bigcup_{\nindex\in \support{\ja{}}}\Group_\nindex$. Therefore, $(\NI\cap \support{\ja{}})\subseteq X_0$. Let $X'=\support{\ja{}}-(\NI\cap \support{\ja{}})$. Clearly, $\support{\ja{}}\subseteq X_0\cup X'$.

  Now we show $\clist(\sigma)$ is satisfiable by consider the following cases:
  \begin{enumerate}[label={Case \arabic*}]
    \item Assume $\clist(\ja{})=\theta \wedge \bigwedge_{\pindex\in Y_1}\neg \psi_\pindex$ and $\ja{}(a)\in \rlist(a)$ for all $a\in \AGT$, where $\theta = \bigwedge_{\nindex\in \support{\ja{}}}\psi_\nindex \wedge \bigwedge_{\pindex \in Y_{0}}\neg \psi_\pindex$. Note $\support{\ja{}}$ is neat. By satisfiability-reduction-con\-di\-tion, $\clist(\sigma)$ is satisfiable.
    \item Assume $\clist(\ja{})=\theta \wedge \psi_{\impeach{\ja{}}}$, $\ja{}(a)\notin \rlist(a)$ for some $a\in \AGT$,
    $\bigcup_{\nindex\in \support{\ja{}}}\Group_\nindex\subseteq \Group_\impeach{\ja{}}$, and $\impeach{\ja{}}\in \PI$. Note $\support{\ja{}}$ is neat. By satisfiabili\-ty-reduction-condition, $\clist(\sigma)$ is satisfiable.
    \item Assume $\clist(\ja{})=\theta \wedge \psi_{\impeach{\ja{}}}$, $\ja{}(a)\notin \rlist(a)$ for some $a\in \AGT$,
    $\bigcup_{\nindex\in \support{\ja{}}}\Group_\nindex\subseteq \Group_\impeach{\ja{}}$, $\ja{}(a)\in \rlist(a)$ for all $a\in \Group_\impeach{\ja{}}$, and $\impeach{\ja{}}\in \PIrat$. Note $\support{\ja{}}$ is neat and $\support{\ja{}}\subseteq X_0\cup X'$ for some $I'\subseteq \NIrat$. By satisfiability-reduction-condition, 
    
    $\clist(\sigma)$ is satisfiable.
    \item Assume $\clist(\ja{})=\theta=\bigwedge_{\nindex\in \support{\ja{}}}\psi_\nindex \wedge \bigwedge_{\pindex \in Y_{0}}\neg \psi_\pindex$. Note $\support{\ja{}}$ is neat. By satisfiability-reduction-con\-di\-tion, $\clist(\sigma)$ is satisfiable.
  \end{enumerate}
 
 Let $\overline{\chi }$ be a conjunction of atomic propositions, %which is
 equivalent to $\neg \chi$. By realization theorem (Theorem \ref{theorem: realization}), there is a pointed model $(P,w_0)$, where $P=( M , {(\preceq_{a,w})_{a\in \AGT, w \in W }})$ and $M=( W, \ACT,\\
 (\relAct{\jactatm} )_{\jactatm \in \JACT},  \valProp )$, such that:
 \begin{itemize}
   \item $P,w_0\vDash \overline{\chi }$;
   \item $\ACT_0\subseteq \ACT$;
   \item for all $\strategymap\in \stratsetatl_{M}^\AGT$ and $\sigma\in\JACT^{\AGT}_0$, if $\strategymap(a)(w_0)=\sigma(a)$ for all $a\in \AGT$, then $M,\lambda^{M,w_0,\strategymap}(1)\vDash \clist(\sigma)$;
   \item for all coalition $\Group$, $\strategymap_{\Group}\in \stratsetatl_{M}^{\Group}$ and $a\in \Group$: $\strategymap_{\Group}|_{\{a\}}\not \in \mathit{Dom}_{M,w_0}^a$ iff $\strategymap_{\Group}(a)(w_0)\in \rlist(a)$.
 \end{itemize}
We are going to show 
\begin{align*}
    P,w_0\vDash \neg \chi \wedge \bigwedge_{\nindex\in \NI}\atlop{\Group_\nindex}\nexttime \psi_\nindex 
 \wedge 
  \bigwedge_{\nindex\in \NIrat}\atloprat{\Group_\nindex}\nexttime \psi_\nindex \\
  \wedge 
  \bigwedge_{\pindex\in \PI}\neg \atlop{\Group_\pindex}\nexttime \psi_\pindex 
  \wedge 
  \bigwedge_{\pindex\in \PIrat}\neg \atloprat{\Group_\pindex}\nexttime \psi_\pindex.
  \end{align*}
  By realization theorem (Theorem \ref{theorem: realization}), $P,w_0\vDash \neg \chi$ clearly holds.

  Let $\nindex\in \NI$. We want to show $P,w_0\vDash \atlop{\Group_\nindex}\nexttime \psi_\nindex $. Let $\strategymap_{\Group_\nindex}\in \stratsetatl^{\Group_\nindex}_{M}$ such that $\strategymap_{\Group_\nindex}(a)(w_0)=\nindex$ for all $a\in \Group_\nindex$. Let $\sigma\in\JACT^{\AGT}_0$ such that $\strategymap(a)(w_0)=\sigma(a)$. Then $\sigma(a)=\nindex$ for all $a\in \Group_\nindex$. Then $\nindex\in \support{\sigma}$. Therefore, $\psi_\nindex$ is a conjunct of $\clist(\sigma)$. Then $M,\lambda(1)\vDash \psi_\nindex$ for all $\lambda\in \outset(w, \strategymap_{\Group_\nindex})$. Therefore, $P,w_0\vDash \atlop{\Group_\nindex}\nexttime \psi_\nindex $. 

  Let $\nindex\in \NIrat$. We want to show $P,w_0\vDash \atloprat{\Group_\nindex}\nexttime \psi_\nindex $.  
  Let $\strategymap_{\Group_\nindex}\in \stratsetatl^{\Group_\nindex}_{M}$ such that 
  $\strategymap_{\Group_\nindex}(a)(w_0)=\nindex$ for all $a\in \Group_\nindex$. Then $\strategymap_{\Group_\nindex}(a)(w_0)\in \rlist(a)$ for all  $a\in \Group_\nindex$. Then $\strategymap_{\Group_\nindex}|_{\{a\}}\not \in \mathit{Dom}_{M,w_0}^a$ for all $a\in \Group_\nindex$.  
  Similar to the case $\nindex\in \NI$, $M,\lambda(1)\vDash \psi_\nindex$ for all $\lambda\in \outset(w, \strategymap_{\Group_\nindex})$. Therefore, $P,w_0\vDash \atlop{\Group_\nindex}\nexttime \psi_\nindex $.

  Let $\pindex\in \PI$. We want to show $P,w_0\vDash \neg \atlop{\Group_\pindex}\nexttime \psi_\pindex $. Let $\strategymap_{\Group_\pindex}\in \stratsetatl^{\Group_\pindex}_{M}$. We want to show $M,\lambda(1)\vDash \neg \psi_\pindex$ for some $\lambda\in \outset(w, \strategymap_{\Group_\pindex})$. Consider the following two cases:

  \begin{enumerate}[label={Case \arabic*}]
  \item Assume $\pindex\in Y_0$. Then $\neg \psi_\pindex$ is a conjunct of $\clist(\sigma)$ for all $\sigma\in\JACT^{\AGT}_0$. Then $M,\lambda(1)\vDash \neg \psi_\pindex$ for some $\lambda\in \outset(w, \strategymap_{\Group_\pindex})$.
  \item Assume $\pindex\notin Y_0$. Then $\Group_\pindex\neq \AGT$. Then there is $a\in \AGT-\Group_\pindex$. Then we claim there is $\sigma_\AGT\in\JACT^{\AGT}_0$ such that $\strategymap_{\Group_\pindex}(b)(w_0)=\sigma_\AGT(b)$ for all $b\in \Group_\pindex$ and $\impeach{\sigma}=\pindex$. 
  Let $\sigma_{\Group_\pindex}:\Group_\pindex\to \ACT_0$ such that $\strategymap_{\Group_\pindex}(b)(w_0)=\sigma_{\Group_\pindex}(b)$ for all $b\in \Group_\pindex$ and $\sigma_{\AGT}:\AGT\to \ACT_0$ such that:
  \begin{itemize}
    \item  $\move{\AGT}{b}=\move{\AGT}{b}$ for all $b\in \Group_\pindex$;
    \item $\move{\AGT}{b}=0$ for all $b \in \AGT-\Group_\pindex$ such that $b\neq a$;
    \item $\move{\AGT}{a}=\begin{cases}
    \pindex-\impeach{\ja{\Group_\pindex}} & \text{if } \pindex \geq \impeach{\ja{\Group_\pindex}} \\
    \pindex-\impeach{\ja{\Group_\pindex}}+n & \text{otherwise}
    \end{cases}$
  \end{itemize}
 Clearly, $\strategymap_{\Group_\pindex}(b)(w_0)=\sigma_\AGT(b)$ for all $b\in \Group_\pindex$ and also $\impeach{\sigma}=\pindex$. Then $\neg \psi_\pindex$ is a conjunct of $\clist(\sigma_\AGT)$. Then there is $M,\lambda(1)\vDash \neg \psi_\pindex$ for some $\lambda\in \outset(w, \strategymap_{\Group_\pindex})$.
  \end{enumerate}
  Therefore, $P,w_0\vDash \neg \atlop{\Group_\pindex}\nexttime \psi_\pindex $.

  Let $\pindex\in \PIrat$. We want to show $P,w_0\vDash \neg \atloprat{\Group_\pindex}\nexttime \psi_\pindex $. Let $\strategymap_{\Group_\pindex}\in \stratsetatl^{\Group_\pindex}_{M}$ such that $\strategymap_{\Group_\pindex}|_{\{a\}}\not \in \mathit{Dom}_{M,w}^a$ for all $a \in \Group_\pindex$. We want to show $M,\lambda(1)\vDash \neg \psi_\pindex$ for some $\lambda\in \outset(w, \strategymap_{\Group_\pindex})$. Consider the following two cases:
  \begin{enumerate}[label={Case \arabic*}]
  \item Assume $\pindex\in Y_1$. Then $\Group_\pindex=\AGT$. Let $\sigma: \AGT \to \ACT_0$ such that $\sigma(a)=\strategymap_{\Group_\pindex}(a)(w)$ for all $a \in \Group_\pindex$. Note $\strategymap_{\Group_\pindex}|_{\{a\}}\not \in \mathit{Dom}_{M,w_0}^a$ for all $a \in \Group_\pindex$. Then $\sigma(a)\in \rlist(a)$ for all $a\in \AGT$. Then $\neg \psi_\pindex$ is a conjunct of $\clist(\sigma)$. Therefore, $M,\lambda(1)\vDash \neg \psi_\pindex$ for some $\lambda\in \outset(w, \strategymap_{\Group_\pindex})$.
  \item Assume $\pindex\notin Y_1$. Then $\Group_\pindex\neq \AGT$. Similar to the case $\pindex\in \PI$. there is $M,\lambda(1)\vDash \neg \psi_\pindex$ for some $\lambda\in \outset(w, \strategymap_{\Group_\pindex})$.
  \end{enumerate}
  Therefore, $P,w_0\vDash \neg \atloprat{\Group_\pindex}\nexttime \psi_\pindex $
\end{proof}

\section{Proof of Lemma \ref{lemma: upward derivability}}
%\section{Upward derivability}
\label{section: upward derivability}
\begin{proof}
  Assume that derivability-reduction-condition with respect to $\phi $ holds. If $\vdash \chi $, then $\vdash \phi $ clearly holds. Now  Assume there is $I\subseteq \NI$ and $I' \subseteq \NIrat $ such that $X\cup X'$ is neat and one of the following is met:
  \begin{itemize}
    \item there is $\pindex\in \PI$ such that $\bigcup_{\nindex\in X\cup X'} \Group_\nindex \subseteq \Group_\pindex$ and $$\vdash \bigwedge_{\nindex\in X\cup X'} \psi_\nindex\to (\psi_\pindex\vee \bigvee_{\pindex'\in Y_0}\psi_{\pindex'});$$
    \item there is $j\in \PIrat$ such that $\bigcup_{\nindex\in I'} \Group_\nindex\subseteq \Group_\pindex$ and $$
    \vdash \bigwedge_{\nindex\in X_0\cup X'} \psi_\nindex \to (\psi_\pindex\vee \bigvee_{\pindex'\in Y_0}\psi_{\pindex'});$$
    \item $\vdash \bigwedge_{\nindex\in X_0\cup X'} \psi_\nindex \to \bigvee_{\pindex\in Y_0\cup Y_1}\psi_{\pindex}$.
  \end{itemize}
  
  Consider the following 3 cases:
  \begin{enumerate}[label={Case \arabic*}]
    \item Assume there is $j\in \PI$ such that $\bigcup_{\nindex\in X\cup X'} \Group_\nindex \subseteq \Group_\pindex$ and $\vdash \bigwedge_{\nindex\in X\cup X'} \psi_\nindex \to (\psi_\pindex\vee \bigvee_{\pindex'\in Y_0}\psi_{\pindex'})$. 
    By \ref{ax:Sup0}, $$\vdash 
    \bigwedge_{\nindex\in X}\atlop{\Group_\nindex}\nexttime \psi_\nindex 
    \to  
    \atlop{\bigcup_{\nindex\in X}\Group_\nindex}\nexttime \bigwedge_{\nindex\in X}\psi_\nindex.$$  
    By \ref{ax:Sup1} and \ref{ax:MR}, 
    \begin{align*}
        \vdash 
    \bigwedge_{\nindex\in X'}\atloprat{\Group_\nindex}\nexttime \psi_\nindex 
    \to  
    \atlop{\bigcup_{\nindex\in X'}\Group_\nindex}\nexttime \bigwedge_{\nindex\in X}\psi_\nindex.
    \end{align*} 
    Therefore, \begin{align*}
    \vdash 
    (\bigwedge_{\nindex\in I}\atlop{\Group_\nindex}\nexttime \psi_\nindex \wedge 
    \bigwedge_{\nindex\in I'}\atloprat{\Group_\nindex}\nexttime \psi_\nindex ) \\
    \to  
    \atlop{\bigcup_{\nindex\in X\cup X'}\Group_\nindex}\nexttime \bigwedge_{\nindex\in I}\psi_\nindex.
    \end{align*}
    By \ref{Mon0}, $\vdash 
    (\bigwedge_{\nindex\in \NI}\atlop{\Group_\nindex}\nexttime \psi_\nindex 
    \wedge 
    \bigwedge_{\nindex\in \NIrat}\atloprat{\Group_\nindex}\nexttime \psi_\nindex)
    \to 
    \atlop{\Group_\pindex}\nexttime (\psi_\pindex\vee \bigvee_{\pindex'\in Y_0}\psi_{\pindex'})$. By \ref{ax:Cro}, \begin{align*}
        \vdash 
    \atlop{\Group_\pindex}\nexttime (\psi_\pindex\vee \bigvee_{\pindex'\in Y_0}\psi_{\pindex'})\to \\ \big (\atlop{\Group_\pindex}\nexttime 
    \psi_\pindex\vee \bigvee_{\pindex'\in Y_0}\atlop{\AGT} \nexttime \psi_{\pindex'}\big ). 
    \end{align*}
    Note $\Group_{\pindex'} = \AGT $ for all $\pindex'\in Y_0$. Therefore, \begin{align*}
    \vdash 
    (\bigwedge_{\nindex\in \NI}\atlop{\Group_\nindex}\nexttime \psi_\nindex 
    \wedge 
    \bigwedge_{\nindex\in \NIrat}\atloprat{\Group_\nindex}\nexttime \psi_\nindex)
    \to 
    \\
    \bigvee_{\pindex'\in \PI } \atlop{\Group_\pindex}\nexttime \psi_{\pindex'}.
    \end{align*}
  
    \item\label{case 2} Assume there is $\pindex\in \PIrat$ such that $\bigcup_{\nindex\in X'} \Group_\nindex\subseteq \Group_\pindex$ and $
    \vdash \bigwedge_{\nindex\in X_0\cup X'} \psi_\nindex 
    \to (\psi_\pindex\vee \bigvee_{\pindex'\in Y_0}\psi_{\pindex'})$.
    Note $\Group_\nindex = \emptyset $ for all $\nindex\in X_0$. By\ref{ax:NP}, $\vdash 
    \bigwedge_{\nindex\in X_0}\atlop{\emptyset }\nexttime \psi_\nindex \to \bigwedge_{\nindex\in X_0}\atloprat{\emptyset }\nexttime \psi_\nindex$.
    By \ref{ax:Sup1}, $$\vdash 
    \bigwedge_{\nindex\in X\cup X_0}\atloprat{\Group_\nindex}\nexttime \psi_\nindex 
    \to  
    \atloprat{\bigcup_{\nindex\in X\cup X_0}\Group_\nindex}\nexttime \bigwedge_{\nindex\in X\cup X_0}\psi_\nindex.$$  
    By \ref{Mon1}, \begin{align*}
     \vdash 
    (\bigwedge_{\nindex\in \NI}\atlop{\Group_\nindex}\nexttime \psi_\nindex 
    \wedge 
    \bigwedge_{\nindex\in \NIrat}\atloprat{\Group_\nindex}\nexttime \psi_\nindex)
    \to \\
    \atloprat{\Group_\pindex}\nexttime (\psi_\pindex\vee \bigvee_{\pindex'\in Y_0}\psi_{\pindex'}).
    \end{align*}
    By \ref{Cro0.5}, \begin{align*}
    \vdash 
    \atloprat{\Group_\pindex}\nexttime (\psi_\pindex\vee \bigvee_{\pindex'\in Y_0}\psi_{\pindex'})\to \\ (\atloprat{\Group_\pindex}\nexttime \psi_\pindex\vee \bigvee_{\pindex'\in Y_0}\atlop{\AGT} \nexttime \psi_{\pindex'}).\end{align*}
    Therefore, \begin{align*}
    \vdash 
    (\bigwedge_{\nindex\in \NI}\atlop{\Group_\nindex}\nexttime \psi_\nindex 
    \wedge 
    \bigwedge_{\nindex\in \NIrat}\atloprat{\Group_\nindex}\nexttime \psi_\nindex )
    \to \\
    (\bigvee_{\pindex'\in \PI}\atlop{\Group_{\pindex'}}\nexttime \psi_{\pindex'}\vee \bigvee_{\pindex'\in \PI}\atloprat{\Group_{\pindex'}}\nexttime \psi_{\pindex'}).
    \end{align*}

    \item Assume $\vdash \bigwedge_{\nindex\in X_0\cup X'} \psi_\nindex \to \bigvee_{\pindex\in Y_0\cup Y_1}\psi_{\pindex}$.
    Similar to \ref{case 2}, we have \begin{align*}
        \vdash 
    (\bigwedge_{\nindex\in \NI}\atlop{\Group_\nindex}\nexttime \psi_\nindex 
    \wedge 
    \bigwedge_{\nindex\in \NIrat}\atloprat{\Group_\nindex}\nexttime \psi_\nindex)
    \to \\
    \atloprat{\AGT}\nexttime \bigvee_{\pindex\in Y_0\cup Y_1}\psi_{\pindex}.
    \end{align*}
    By \ref{ax:DGRA},  \begin{align*}
        \vdash 
    \atloprat{\AGT}\nexttime \bigvee_{\pindex\in Y_0\cup Y_1}\psi_{\pindex}
    \to \bigvee_{\pindex\in Y_0\cup Y_1}\atloprat{\AGT}\nexttime \psi_{\pindex}.
    \end{align*}
    Note $\Group_{\pindex} = \AGT $ for all $\pindex\in Y_0\cup Y_1$. 
    Therefore, \begin{align*}
    \vdash 
    (\bigwedge_{\nindex\in \NI}\atlop{\Group_\nindex}\nexttime \psi_\nindex 
    \wedge 
    \bigwedge_{\nindex\in \NIrat}\atloprat{\Group_\nindex}\nexttime \psi_\nindex)
    \to \\ 
    (\bigvee_{\pindex\in Y_0}\atlop{\AGT}\nexttime \psi_{\pindex}\vee \bigvee_{\pindex\in Y_1}\atloprat{\AGT}\nexttime \psi_{\pindex}).
    \end{align*}
  \end{enumerate}
  Therefore, \begin{align*}
      \vdash 
    (\bigwedge_{\nindex\in \NI}\atlop{\Group_\nindex}\nexttime \psi_\nindex 
    \wedge 
    \bigwedge_{\nindex\in \NIrat}\atloprat{\Group_\nindex}\nexttime \psi_\nindex)
    \to \\
    (\bigvee_{\pindex\in \PI}\atlop{\Group_\pindex}\nexttime \psi_{\pindex}\vee \bigvee_{\pindex\in \PIrat}\atloprat{\Group_\pindex}\nexttime \psi_{\pindex}).
    \end{align*}
    Then $\vdash \phi $.
  \end{proof}

 \section{Proof of Theorem \ref{thm:soundnessCL}}\label{section: complete proof}
  \begin{proof}
Proving soundness is a
routine exercise. 
 We just show  completeness
 by an induction on the modal
 degrees of formulas. 
 The $0$-modal degree case
 follows from
 completeness of classical propositional logic. 
Let 
   $\phi \in \lang_{\clratlogic}(\ATM, \AGT )$ with modal degree  greater than $0$. Assume $ \phi $ 
   is valid
   and for any 
    $\psi \in \lang_{\clratlogic}(\ATM, \AGT )$, if
    the modal degree of $\psi $ is lower
    than
    the modal degree of
    $\phi $ and $ \psi $ is valid, then $\vdash \psi $. We want to show that 
    $\vdash \phi $.
  By the normal-form lemma (Lemma \ref{lemma: normal-form}), $\phi \Liff \bigwedge_{k\in K}\phi_k$
  is valid, where $K$ is a set of indices and $\phi_k $ is a standard
  $ \clratlogic$ disjunction whose modal degree is not greater than $\phi $. Fix $k\in K$. Then $ \phi_k$
  is valid. By the downward validity lemma (Lemma \ref{lemma: downward validity}),  the validity-reduction condition for 
  $\phi_k $ holds. By inductive hypothesis, the 
  derivability-reduction condition
  for $\phi_k $ holds. By 
  the upward derivability lemma (Lemma \ref{lemma: upward derivability}), $\vdash \phi_k$. Therefore, $\vdash \bigwedge_{k\in K}\phi_k$. Note that $\vdash \phi \Liff \bigwedge_{k\in K}\phi_k$. Thus,  $\vdash \phi $.
\end{proof}
  
  \section{Proof of Theorem \ref{thm:modelchecking}}\label{section: proof for modelchecking}

  \begin{proof}

%Before defining the model-checking problem for $\clratlogic$
We start by introducing 
some preliminary notions. In Section \ref{sec:pref}
an agent's preference relation at a given state was
defined relative to the set of computations
starting in this state. 
%It turns out that 
In the case
of a GCSP with short-sighted preferences
we can replace 
each preference relation $     \preceq_{i,w} $
by  a partial preorder 
\begin{align*}
     \preceq_{i,w}^{\mathit{succ} }  \subseteq \mathcal{R}(w) \times \mathcal{R}(w)    
\end{align*}
defining agent $i$'s preference relation over the successors of state $w$.
Then, we can define the notion of agent $i$'s dominated action
at a given state $w$. Specifically, 
we say that at world $w$
  agent $i$'s action  $\delta_{\{ i \}} \in  \JACT_{  \{i \} }   $
  dominates agent $i$'s  action  $\delta_{\{ i \}}' \in  \JACT_{  \{i \} }   $
  iff 
  \begin{align*}
  \forall \delta_{\AGT \setminus \{i \}} & \in \JACT_{\AGT  \setminus \{i \}}, \\  
& 
\text{if }  \mathcal{R}_{ \delta_{\AGT \setminus \{i \}} \oplus  \delta_{\{i \}}  } (w) \neq
\emptyset \text{ and }
 \mathcal{R}_{ \delta_{\AGT \setminus \{i \}} \oplus  \delta_{\{i \}'}  ' } (w) \neq
\emptyset \text{ then } \\
&
\mathit{succ}(w, \delta_{\AGT \setminus \{i \}} \oplus \delta_{\{i \}} )
  \prec_{i,w}^{\mathit{succ} } 
\mathit{succ} ( w, \delta_{\AGT \setminus \{i \}} \oplus  \delta_{\{i \}} ')  ,
  \end{align*}
where $\mathit{succ}(w, \delta_{\AGT \setminus \{i \}} \oplus  \delta_{\{i \}} ) $
is the unique element of the transition relation 
  $\mathcal{R}_{ \delta_{\AGT \setminus \{i \}} \oplus  \delta_{\{i \}}  } (w) $,
 $\mathit{succ}(w, \delta_{\AGT \setminus \{i \}} \oplus  \delta_{\{i \}}' ) $
is the unique element of 
  $\mathcal{R}_{ \delta_{\AGT \setminus \{i \}} \oplus  \delta_{\{i \}}'  } (w) $,
  and for any coalition $\Group$, 
  $\JACT_\Group   $
  is the set of joint actions for $\Group$,
  i.e., the set of all functions $\delta_\Group : \Group \longrightarrow \ACT$. 
   We denote by $\mathit{ActDom}_{M,w}^i$ the set of dominated actions of agent $i$   at $w$.

  We can now reinterpret the modalities $\atloprat{\Group}\nexttime$, $\atloprat{\Group}\henceforth$, and $\atloprat{\Group}\until{}{}$
by replacing the set $\mathit{Dom}_{M,w}^i$ by 
$\mathit{ActDom}_{M,w}^i$. 
It is routine to show that the two semantic interpretations, the one defined with $\mathit{Dom}_{M,w}^i$ and the one with $\mathit{ActDom}_{M,w}^i$ are equivalent for the case of GCSP with short-sighted preferences. 

%%%

\iffalse
Given  a CGS $M  = ( W,  \ACT,  \relActJoint,  \valProp  )$
and a
goal base profile 
$\Gamma= (G_1, \ldots , G_n)$
with 
$G_i\subseteq
\mathcal{L}_{\mathsf{LTL}}$
for every agent $i\in \AGT$,
we can induce the corresponding CGS with stable preferences
$(M, \Omega_M^\Gamma)$
with $\Omega_M^\Gamma=(\preceq_{i,w }^\Gamma   )_{i \in \AGT, w \in W }$
as follows: 
\begin{itemize}
    \item for every $i\in \AGT$,
    for every $w \in W$
    \begin{align*}
            \preceq_{i,w }^\Gamma = &
    \Big\{ 
    (\lambda', \lambda) \in
    \historyset_{M} \times
    \historyset_{M} \suchthat  \\
&    \mathit{card}\big(
\{   \varphi \in G_i: (M,\lambda') \models \varphi\}
    \big) \geq
     \mathit{card}\big(
\{   \varphi \in G_i: (M,\lambda) \models \varphi\}
    \big)
    \Big\}.
    \end{align*} 
\end{itemize}
\fi

Now we turn to the global model checking problem for \atlratlogic. 
We show that the complexity of the model checking problem is \Ptime-complete. 
The lower-bound follows from the model
checking of  \atllogic which is \Ptime-complete \cite{alur2002alternating}. % (notice that model checking \cllogic is also Ptime  \cite{pauly2001logic}). 

For the upper bound, we provide a model-checking algorithm. 
%We denote by $\stratsetatl_r$ the set of memoryless strategies.  
Given a CGSP $P = (M,\Omega_M)$, with   $M  = ( W,  \ACT$, $ \relActJoint, \allowbreak  \valProp  )$  and $\Omega_M=(\preceq_{i,w }   )_{i \in \AGT, w \in W }$, and a formula $\varphi$, Algorithm \ref{alg:mc} works similarly to the symbolic algorithm for \atllogic model checking \cite{alur2002alternating} by 
computing the set of states in which  $\varphi$ holds by recursively considering its subformulas,  
%As Algorithm modelCheck visits each subformula at most once, and the number of subformulas is not greater than the size of , the algorithm can clearly be implemented in

\iffalse
  The model-checking problem can be formulated as follows:
\begin{itemize}
    \item 
    \textbf{INPUT}:  
    A CGSP $P = (M,\Omega_M)$, where  $M  = ( W,  \ACT$, $ \relActJoint,  \valProp  )$ is a CGS and $\Omega_M=(\preceq_{i,w }   )_{i \in \AGT, w \in W }$ is a preference structure,  
    and a formula $\varphi$.

    \item 
    \textbf{OUTPUT}:  $\{s: (P,s) \models \varphi, s\in W\}$, i.e., the set of states  where the $\varphi$ holds in $P$
    %    \textbf{OUTPUT}: yes if $(P^\Gamma,s) \models$    with $P^\Gamma=(M, \Omega_M^\Gamma) $, no otherwise. 
\end{itemize}
\fi
%

%Given a  \clratlogic formula $\varphi$ and a CGSP $P$,  Algorithm \ref{alg:mc} returns all the states in which  $\varphi$ holds.

\begin{breakablealgorithm}
\label{alg:mc}
	\caption{Model checking R-CL } 
	\begin{algorithmic}[1] 
    %\hskip\algorithmicindent 
   \noindent \textbf{Input: }   a CGSP $P = (M,\Omega_M)$, with   $M  = ( W,  \ACT$, $ \relActJoint, \allowbreak  \valProp  )$  and $\Omega_M=(\preceq_{i,w }   )_{i \in \AGT, w \in W }$, and a \atlratlogic formula $\varphi$.
   \par 
   \noindent \textbf{Output: } $\{w: (P,w) \models \varphi, w\in W\}$ %, i.e., the set of states  where the $\varphi$ holds in $P$
		\Procedure{MC}{$P,  \varphi$}		
      \Case {$\varphi = p$}
            \State{\textbf{return} $\{ w : p\in \valProp(w)$\} }   
        \EndCase
    
        \Case {$\varphi = \lnot \psi$}
            \State{\textbf{return}  $W \setminus \textsc{MC} (P,   \psi)$  }
        \EndCase
       \Case {$\varphi = \psi \land \chi$}
            \State{\textbf{return} $\textsc{MC} (P,\psi) \cap \textsc{MC} (P,\chi)$}
        \EndCase
        %Next
      \Case {$\varphi =  \atlop{\Group} \nexttime \psi$}
             \State{\textbf{return} $ \mathit{Pre}(P, C, \textsc{MC}(P,  \psi))$}
        \EndCase
        %Globally
        \Case {$\varphi =  \atlop{\Group} \henceforth \psi$}
            \State{$\rho \gets W, \tau \gets MC(P, \psi)$}
            \While{$\rho \nsubseteq \tau$}
             \State{$\rho \gets \tau$, $\tau \gets \mathit{Pre}(\Group, \rho) \cap MC(P,\psi)$}
             \EndWhile
             \State{\textbf{return} $\rho$}
        \EndCase
        %Until
        \Case {$\varphi =  \atlop{\Group} \until{\psi_1}{\psi_2}$}
            \State{$\rho \gets \empty, \tau \gets MC(P, \psi_2)$}
            \While{$\tau \nsubseteq \rho$}
             \State{$\rho \gets \rho \cup \tau$, $\tau \gets \mathit{Pre}(\Group, \rho) \cap MC(P,\psi_1)$}
             \EndWhile
             \State{\textbf{return} $\rho$}
        \EndCase
        %RatNext
        \Case {$\varphi =  \atloprat{\Group} \nexttime \psi$}
             \State{\textbf{return} $ \mathit{PreRat}(P, C, \textsc{MC}(P,  \psi))$  }
        \EndCase
        %RatGlobally
        \Case {$\varphi =  \atloprat{\Group} \henceforth \psi$}
            \State{$\rho \gets W, \tau \gets MC(P, \psi)$}
            \While{$\rho \nsubseteq \tau$}
             \State{$\rho \gets \tau$, $\tau \gets \mathit{PreRat}(\Group, \rho) \cap MC(P,\psi)$}
             \EndWhile
             \State{\textbf{return} $\rho$}
        \EndCase
        %RatUntil
         \Case {$\varphi =  \atloprat{\Group} \until{\psi_1}{\psi_2}$}
            \State{$\rho \gets \empty, \tau \gets MC(P, \psi_2)$}
            \While{$\tau \nsubseteq \rho$}
             \State{$\rho \gets \rho \cup \tau$, $\tau \gets \mathit{PreRat}(\Group, \rho) \cap MC(P,\psi_1)$}
             \EndWhile
             \State{\textbf{return} $\rho$}
        \EndCase
    \EndProcedure
	\end{algorithmic}
\end{breakablealgorithm}

To solve formulas in the form $\atlop{\Group} \nexttime \psi$, $\atlop{\Group} \henceforth \psi$, and $\atlop{\Group} \until{\psi_1}{\psi_2}$, the algorithm makes use of the function $\mathit{Pre}(P, \Group, Q)$. The function  takes as arguments a CGSP $P$, a coalition $\Group \subseteq \AGT$, and a set $Q \subseteq S$ and returns the set of states from which coalition $\Group$ can force the outcome to be in one of the $Q$ states. 
Formally, 

\begin{align*}\mathit{Pre}(P, \Group, Q) = & \{ w : w\in W, \exists \delta_{\Group} \in \JACT_{\Group}, \\
& \forall \delta_{\AGT\setminus\Group}\in \JACT_{\AGT\setminus\Group},  \relAct{\delta_{\Group} \oplus    \delta_{\AGT\setminus\Group}}(w)\subseteq Q \}
\end{align*}

%The main difference is in computing the operator 
To handle formulas with rational strategic operators, we first need to compute agents' dominated actions.  
%For coalition logic, strategies are action profiles. 
Deciding whether an action is dominated requires checking, for every other action for that player, whether the latter action performs strictly better against all the opponents’  actions. This procedure is known to be in polynomial time 
\cite{conitzer2005complexity,gilboa1993complexity}. 
%Given,  (for a given CGSP $P$, agent $i \in \AGT$ and state $w$) 
For computing the set of actions for agent $i$ that are dominated from the state $w$, we consider the function $\mathit{DomAct}(P, i, w)$, defined as follows: 

\begin{align*}
\mathit{DomAct}&(P, i, w) = \{ \delta_{\{i\}} \in \JACT_{\{i\}} : \exists \delta_{\{i\}}' \in \JACT_{\{i\}}  , \\ &  \forall \delta_{\AGT\setminus\{i\}} \in \JACT_{\{i\}}^{\AGT\setminus\{i\}},   \\ & \mathit{succ}(w, \delta_{\{i\}}\oplus  \delta_{\AGT\setminus\{i\}} ) \prec_{i,w}^{\mathit{succ}}   \mathit{succ}(w, \delta_{\{i\}}'\oplus  \delta_{\AGT\setminus\{i\}} ) \}   \end{align*}

The non dominated joint actions of coalition $\Group$ in $w$ are obtained by function $\mathit{NonDomJAct}(P, i, w)$, defined as follows: 

\begin{align*}
\mathit{NonDomJAct}(P, \Group, w) = & \{ \delta_\Group \in \JACT_\Group:    \forall i \in \Group, \\
& \delta_{\{i\}}  \not \in DomAct(P,i,w)\}   \end{align*}
%%Remove \stratsetatl_r 

%\begin{align*}
%Dom(P, i, w) =& \{ a \in \ACT^{\AGT }  : \exists a' \in  \ACT^{\AGT }  , \forall a_{-i} \in \ACT^{\AGT\setminus\{i\}},   \\ & \outset(w, \sigma_i\oplus  \sigma_{-i} ) \prec_{i,w }  s' \}  
%\end{align*}
%$\relAct{a_{\Group} \oplus    a_{\AGT\setminus\Group}}(s)$

% \forall \lambda \in \outset(w, \strategymap_\Group), \big( P, \lambda(1)  \big)\models \varphi ,
For the cases of formulas in the form $\atloprat{\Group} \nexttime \psi$, $\atloprat{\Group} \henceforth \psi$,  and $\atloprat{\Group} \until{\psi_1}{\psi_2}$, the algorithm uses the function $\mathit{PreRat}(M, \Group, Q)$ that computes the states from where the coalition $\Group$ can force   the outcome to be in one of the $Q$ states while playing non dominated actions.
%computes for a given CGSP $P$, coalition $\Group \subseteq \AGT$ and a set $Q \subseteq S$, the set of states, from which coalition $\Group$ can force the outcome to be in one of the $Q$ states.
\begin{align*}
    \mathit{PreRat}(P, \Group, Q)  = & \{ w : w\in W, \exists \delta_\Group \in \mathit{NonDomJAct}(P,\Group,w), \\
%\exists \delta_{\Group} \in \JACT_{\Group} \textcolor{blue}{\setminus Dom(P, i, w)}, \\ 
& %\delta_{\Group} = (\delta_{\{i\}}(w))_{i \in \Group}, 
\forall \delta_{\AGT\setminus\Group}\in \JACT_{\AGT\setminus\Group}, \\ 
& \relAct{\delta_{\Group} \oplus    \delta_{\AGT\setminus\Group}}(w)\subseteq Q   \}
\end{align*} 

Notice that functions $Pre$ and $PreRat$ can
be computed in polynomial time. 
Algorithm \ref{alg:mc}'s complexity is in $\mathcal{O}(|M| \cdot l)$, where $l$ denotes the length of the input formula $\varphi$. The outer loop of the algorithm depends on l, and it can be shown that each one of the cases lies in $\mathcal{O}(|M|)$. 

\end{proof}